\documentclass{CSML}
\pdfoutput=1

\usepackage{lastpage}
\newcommand{\dOi}{14(3:11)2018}
\lmcsheading{}{1--\pageref{LastPage}}{}{}%
{Sep.~29,~2015}{Aug.~27, 2018}{}

\newif\iflong \longfalse

\usepackage {amsthm,amssymb,amsbsy}

\input{paper.sty}

\newcommand{\DSb}{}
\newcommand{\DSe}{}

\newcommand{\Db}{}
\newcommand{\De}{}

\newcommand{\XXb}{}
\newcommand{\XXe}{}

\newcommand{\Xb}{}
\newcommand{\Xe}{}

\newcommand{\clx}[1]{}
\newcommand{\Xab}{}
\newcommand{\Xae}{}
\newcommand{\XXab}{}
\newcommand{\XXae}{}

\newcommand{\crossDS}{{\tt{X}}}

\renewcommand{\equiv}{=}

\usepackage{etex}
\usepackage{verbatim}
\usepackage{fancyhdr}
\usepackage{fancybox}
\usepackage{stmaryrd}
\usepackage[reqno]{amsmath}
\usepackage{mathtools}
\usepackage{mathpartir}
\usepackage{extarrows}
\usepackage{amssymb}
\usepackage{amsfonts}
\usepackage{xypic}\xyoption{all}
\usepackage{tikz}

\usepackage {indentfirst}
\usepackage{graphicx}
\usepackage[hang]{subfig}

\usepackage{cite}

\usepackage{proof}

\usepackage{bbding}
\usepackage{hyperref} 
\hypersetup{
    bookmarks={true},         
    bookmarksnumbered={true},
    bookmarksopen={true},
    unicode={true},         
    pdftoolbar={true},        
    pdfmenubar={true},        
    pdffitwindow={false},     
    pdfstartview={FitH},    
    pdftitle={My title},    
    pdfauthor={Author},     
    pdfsubject={Subject},   
    pdfcreator={Creator},   
    pdfproducer={Producer}, 
    pdfkeywords={keyword1, key2, key3}, 
    pdfnewwindow={true},      
    colorlinks={true},       
    linkcolor={blue},          
    citecolor={cyan},        
    anchorcolor={green},
    filecolor={magenta},      
    urlcolor={black},           
    hidelinks
}
\newcommand{\para}{|\,}

\newcommand{\fosub}[2]{\{#1/#2\} }
\newcommand{\hosub}[2]{\{#1/#2\} }

\newcommand{\vect}[1]{{#1_1,#1_2,...,#1_n} }
\newcommand{\ve}[1]{\widetilde{#1}}
\newcommand{\seq}[1]{\ve{#1}}  %

\newcommand{\stm}[1]{{\,\xrightarrow{#1}} }
\newcommand{\wt}[1]{{\,\xLongrightarrow{#1}} }

\newcommand{\rc}[1]{{\color{red} #1}}
\newcommand{\bc}[1]{{\color{blue} #1}}
\newcommand{\da}{\!\!\downarrow}
\newcommand{\Da}{\!\!\Downarrow}

\newcommand{\sep}{\vspace*{0.7cm}}
\newcommand{\sepp}{\vspace*{0.3cm}}

\newcommand{\nsepv}[1]{\vspace{0mm}}

\newcommand{\rb}[1]{\raisebox{3.2ex}[0pt]{#1}}

\newcommand{\lds}{[\![}
\newcommand{\rds}{]\!]}
\newcommand{\encodingm}[3]{ \lds #1 \rds^{#2}_{#3}}

\newcommand{\backs}[1]{\!\setminus\!#1\!}

\newcommand{\DEF}{\stackrel{{\rm{def}}}{=}} 

\newcommand{\lrangle}[1]{\langle #1 \rangle} 

\newcommand{\shr}{\mathrel{\stm{}_{{\rm h}}}} 
\newcommand{\whr}{\mathrel{\wt{}_{{\rm h}}}} 

\newcommand{\BISI}{\approx  } 

\newcommand{\BSTA}{\approx_{\rm{bb}}^{\rm{lin,asy}} } 
\newcommand{\BCTA}{\approx_{\rm{bc}}^{\rm{lin,asy}} } 
\newcommand{\CTTA}{\approx_{\rm{ct}}^{\rm{lin,asy}} } 

\newcommand{\subtp}{\trianglelefteq }  
\newcommand{\TPEQ}{{\sim_{type}} } 

\newcommand{\contrdiv}{^{\Uparrow}\!\!\!\contr} 
\newcommand{\conref}[1]{{(\ref{#1})}} 
\newcommand{\defref}[1]{{\ref{#1}}} 
\newcommand{\hardcode}[1]{} 
\newcommand{\supp}[1]{\mathbf{supp}(#1) }
\newcommand{\afig}{Figure } 

\newcommand{\xxstress}[1]{{#1}}
\newcommand{\xx}[1]{}

\newenvironment{xxenv}{\!\!}{\!\!}
\newcommand{\xxx}[1]{#1}

\newcommand{\remxx}[1]{{\Large \xx{#1}}}
\newif\ifxxremark \xxremarkfalse 

\thicklines

\newcommand{\finish}[1]{}
\newcommand{\uptoc}{up-to-$\leq$-and-contexts}
\newcommand{\uptoE}{up-to-$\expa$-and-contexts}
\newcommand{\uptoEdiv}{up-to-$\expa^\Div$-and-contexts}
\newcommand{\expaDiv}{\mathrel{\expa^\Div}}
\newcommand{\contrDiv}{\mathrel{\mbox{}^\Div\!\!\contr}}

\newcommand{\leqNAT}{\leqslant}
\newcommand{\geqNAT}{\geqslant}

\def\RR{\mathrel{\mathcal{R}}} 

\theoremstyle{plain}

\newtheorem{theorem}{Theorem}[section]
\newtheorem{proposition}[theorem]{Proposition}
\newtheorem{lemma}[theorem]{Lemma}

\theoremstyle{definition}
\newtheorem{definition}[theorem]{Definition}
\newtheorem{example}[theorem]{Example}

\theoremstyle{remark}

\begin{document}

\title{Trees from Functions as Processes}

\author{Davide Sangiorgi}
\address[]{Universit{\`a} di Bologna (Italy) and INRIA (France) }
\author{Xian Xu}
\address{East China University of Science and Technology (China)}

\thanks{This work has been supported by project ANR 12IS02001 `PACE', NSF of China (61261130589),  and partially supported by NSF of China (61702334, 61772336, 61572318, 61472239, 61872142).} 

\keywords{...} 
\subjclass{.... } 

\titlecomment{The paper is an expanded version of work presented at the CONCUR conference, LNCS 8704:78-92, 2014.}

\begin{abstract}
\LLTN s and B\"{o}hm
Trees are the best known tree structures on the $\lambda$-calculus.
We give general conditions under which an encoding of the
$\lambda$-calculus into the $\pi$-calculus is sound and complete with
respect to
such trees.
We  apply these conditions to various encodings of the call-by-name
$\lambda$-calculus, showing how
the two kinds of  tree can be obtained by varying the behavioural
equivalence adopted in the $\pi$-calculus and/or the encoding. 
\end{abstract}

\maketitle


\section{Introduction}\label{sec-intro}

The $\pi$-calculus is a well-known model of computation with
processes. Since its introduction, its comparison with the
$\lambda$-calculus has received a lot of attention.
Indeed,
a deep comparison  between a process calculus and
 the   $\lambda$-calculus  is interesting for several reasons:
it is  a
significant
test of   expressiveness,  and    helps  in
 getting  deeper insight  into
its theory.
From the $\lambda$-calculus perspective,
it  provides the means  to study $\lambda$-terms
 in    contexts  other  than  purely  sequential ones,
 and with the  instruments available  in
the process calculus.
A more practical motivations for   describing functions as processes
is  to
provide  a semantic foundation  for  languages which  combine
concurrent and functional programming  and to develop   parallel
implementations of functional languages.

Beginning with Milner's seminal work \cite{Mil92}, a number of
 $\lambda$-calculus strategies have been encoded into the
$\pi$-calculus, including call-by-name, strong call-by-name (and
call-by-need variants),  call-by-value, parallel call-by-value (see \cite[Chapter~15]{SW01a}).
In each case, several variant
 encodings have  appeared, by
varying the target language or  details of the encoding itself,
\DSb
see \cite[Part VI]{SW01a} for
details.
\DSe
Usually, when an encoding is given,
 a few  basic results about its correctness are
 established, such as  operational correctness and validity
of reduction (i.e., the property that the encoding of a $\lambda$-term
and the encoding of a 
\xxx{reduct}
of it are behaviourally undistinguishable).
Only in a few cases the question of the equality on $\lambda$-terms
induced by the encoding
 has been tackled, e.g.,
 \cite{San95lazy,San93d,SW01a,BergerHY05,DemangeonH11,BergerHY01};
\DSb
in \cite{San95lazy,SW01a} for encodings of call-by-name and with
respect to the ordinary  bisimilarity of the $\pi$-calculus, in
 \cite{BergerHY05,DemangeonH11,BergerHY01} for various forms of
$\lambda$-calculi, including  polymorphic ones,
and with
respect to
  contextual forms of
 behavioural equivalence    enhanced with types so to
 obtain coarser relations.

In this paper, we refer to the above question
 \DSe
as the \emph{full abstraction} issue:
 for an encoding $\qenco$  of the $\lambda$-calculus into
 $\pi$-calculus,
an equality $=_\lambda$
on the $\lambda$-terms, and
 an equality   $=_\pi$ on the $\pi$-terms,
full abstraction is achieved when  for all $\lambda$-terms $M,N$ we have
$ M  =_\lambda  N $ iff
$\encoding M  =_\pi \encoding N $.
Full abstraction
 has two parts: soundness, which is the implication from
right to left, and
 completeness, which is its converse.

The equality
 $=_\lambda$ usually is not
 the ordinary Morris-style contextual equivalence on the
$\lambda$-terms:
the
 $\pi$-calculus   is richer --- and hence more
discriminating --- than the $\lambda$-calculus; the latter is purely sequential,
whereas the former can also express  parallelism and
non-determinism.  Exception to this are encodings into forms
of $\pi$-calculus equipped with rigid constraints, e.g., typing constraints,
which limit the set of legal $\pi$-calculus contexts \cite{BergerHY05,DemangeonH11,BergerHY01}.

Indeed, the interesting question here is understanding what
 $=_\lambda$ is when
$=_\pi$ is a well-known behavioural equivalence on $\pi$-terms.
This question
   essentially amounts to
 using the
encoding in order to build a $\lambda$-model, and then understanding
the $\lambda$-model itself.
While seldom tackled,
the outcomes of
this study
have been significant: for a few call-by-name encodings
 presented in \cite{SW01a}
it has been shown that, taking (weak) bisimulation on the
$\pi$-terms, then
 $=_\lambda$
 corresponds
to  a well-known tree structure in the $\lambda$-calculus theory, namely
the
 {\em \LLTN s} (LTs) \cite{SW01a}.

There is however another kind of tree structure in the
$\lambda$-calculus, even more important:
the  {\em \BTN s}
 (BTs).
BTs play a  central role in the classical theory of the
$\lambda$-calculus. The  local structure
of some of the most  influential  models of the
$\lambda$-calculus\index{lambda-calculus@$\lambda$-calculus!model of}, like Scott and Plotkin's
$P_\omega$ \cite{Sco76}, Plotkin's $T^\omega$ \cite{Plo78},
is precisely
the BT equality; and the local structure of Scott's
$D_\infty$
(historically
 the first
mathematical, i.e., non-syntactical,
  model of the untyped $\lambda$-calculus) is the equality of the
\begin{xxenv}
`infinite $ \eta$ expansions' of BTs.
\end{xxenv}
\DSb
Details on these and other models of the $\lambda$-calculus
can be found in  the comprehensive books \cite{Bar84,HiSe86}.
The full abstraction results in 
the literature for encodings of $\lambda$-calculus into $\pi$-calculus,
however, only concern LTs \cite{SW01a}.
\DSe

A major reason for the limited attention that the full abstraction
issue for encodings of $\lambda$-calculus into $\pi$-calculus has
received is that understanding what kind of
the structure the encoding produces may be difficult, and
the full abstraction proof itself is long and tedious.
The contribution of this paper is twofold:
\begin{enumerate}
\item We present general conditions for soundness and completeness
of an encoding of the $\lambda$-calculus  with respect to {both} LTs
\emph{and} BTs.
The conditions can be used both on coinductive equivalences such as
bisimilarity, and on  contextual equivalences such as may and
must equivalences \cite{San12a}.

\item We show that by properly tuning the notion of observability
  and/or the details of the encoding it is possible to recover
BTs in place of LTs.
\end{enumerate}
Some conditions only concern   the behavioural equivalence chosen for
the $\pi$-calculus, and are independent of the encoding; a few
conditions are purely syntactic (e.g., certain encoded  contexts
should be  guarded); the only behavioural conditions are equality of
$\beta$-convertible terms, equality among certain unsolvable terms, and
existence of an inverse for certain contexts resulting from the
encoding (i.e., the possibility of extracting their immediate subterms,
up-to the behavioural equivalence chosen in the $\pi$-calculus).
We use these properties to derive full abstraction results for BTs and
LTs for various encodings and various behavioural equivalence of the
$\pi$-calculus. For this we exploit a few basic properties of the
encodings,  making  a large reuse of proofs.

In the paper we  use the conditions with the
 $\pi$-calculus, but potentially they
  could also be used in other concurrency formalisms.

\noindent
\emph{Structure of the paper.} Section~\ref{sec-background}
collects background material. Section~\ref{sec-encoding-notion}
introduces the notion of encoding of the $\lambda$-calculus, and
concepts related to this. Section~\ref{ss:sounds_completes} presents
the conditions for soundness and completeness.
 Section~\ref{sec-examples-cbn} and Section~\ref{sec-examples-scbn} applies the conditions on a few encodings of
 call-by-name and  strong call-by-name from the literature, and for
 various behavioural equivalences on the $\pi$-calculus.
 Section~\ref{s:typ_asy} briefly discusses refinements of the
 $\pi$-calculus, notably with linear  types. Some conclusions are
 reported in Section~\ref{sec-conclusion}.


%
\section{Background}
\label{sec-background}

\subsection{The $\lambda$-calculus }

We use $M,N$ to range over
the set $\Lambda$ of $\lambda$-terms, and              
$x,y,z$ to range over 
variables. 
The syntax of $\lambda$-terms, and the rules for call-by-name and
strong call-by-name (where reduction may continue underneath a
$\lambda$-abstraction)
are standard \cite{Bar84}. 
The set $\Lambda$ of $\lambda$-terms is given by the grammar: 
\[M ::= x \midd \lambda x. M \midd MN 
\]
We will encode call-by-name $\lambda$-calculus, in its weak or strong
form. In both cases, we have  rules $\beta$ and $\mu$, only in the
strong case we have also $\xi$: 
\[
\begin{array}{lll}
\infer[\beta]{(\lambda x.M)M' \stm{} M\hosub{M'}{x}}{} &\qquad \infer[\xi]{\lambda x. M\stm{} \lambda x. M'}{M\stm{} M'} & \qquad \infer[\mu]{MN\stm{} M'N}{M\stm{} M'} 
\end{array}
\]

\DSb
We sometimes omit $ \lambda $ in nested abstractions, \Xab thus for example, $\lambda x_1 x_2. M$ \Xae stands
for   $ \lambda x_1. \lambda  x_2. M$. 
\DSe
We assume the standard concepts of free and bound variables and
substitutions, and
 identify $\alpha$-convertible terms.
\DSb
Thus, throughout the paper `$=$' is syntactic equality modulo
$\alpha$-conversion.
\DSe
We write $\Omega $ for the  purely divergent term 
$(\lambda  x . x x)(\lambda  x . x x )$. 
We  sometimes use $\ve{.}$ for a tuple of elements; for instance 
 $\lambda\ve{x}.M$ stands for $\lambda x_1...x_n.M$ 
 and $\ve{M}$ for $M_1M_2\cdots M_n$,  for some $n$. We write $|\ve e |$ for the
 cardinality of the  tuple $\ve e$, and ${\ve e}_i$ for 
 the $i$-th component of the tuple.

\DSe

\DSb
In order to define L\'{e}vy-Longo trees and B\"ohm trees, 
 we need the notions of  \emph{solvability}, and
of  \emph{head reduction}, which we now introduce (see \cite{DezG01} for a thorough tutorial on such trees).
 We use $n$ to range over the set 
 of   non-negative   integers and $\omega$ to represent the first 
limit ordinal.

A $\lambda$-term is either  of the form  $\lambda \widetilde x . y \widetilde M$ 
or of the form   $\lambda \widetilde x .  (\lambda x. M_0) M_1 \ldots M_n $, \Xab $n \geqNAT 1$. \Xae 
In the latter,  the redex $(\lambda x. M_0) M_1 $  is called the {\em head
redex}. If $M$ has a head redex, then $M\shr N$   holds
  if $N$
results from $M$ by $\beta$-reducing its head redex.
{\em Head reduction},   $\whr$, is the
reflexive and transitive closure of $\shr$.
Head reduction is different from the call-by-name
  reduction ($\Longrightarrow$): a call-by-name 
redex is also a head redex, but  the converse is false
as a  head redex can  also be located  underneath an
abstraction.
The terms 
of the form 
$\lambda \widetilde x . y \widetilde M$, that is,  the terms that cannot be head-reduced, 
are the \emph{head normal forms}.
Since  head reduction is deterministic, the head normal form for a term $M$, 
that is, a head normal form $N$ such that $M \whr N$, 
\XXb if it exists, \XXe
 is unique. The terms that have head normal forms are the \emph{solvable}
terms. The remaining terms are called \emph{unsolvable}. These are the    terms in which
\Db
 head reduction never terminates. It may be however that head reductions on an
\Db unsolvable term \De
 uncover some abstractions. The number of such abstraction defines the
\emph{order of unsolvability} for that term. Formally, an unsolvable term $M$ has \emph{order of
  unsolvability $n$}, for $0\leq n < \omega$ if $n$ is the largest integer such that  
$M \whr \lambda \widetilde x .  M$, for some $M$ and $\widetilde x$ with $ |\widetilde x| = n$; the  
 unsolvable $M$ has \emph{order of
  unsolvability $\omega$} if  for all $n \geq 0 $ we have
$M \whr \lambda \widetilde x .  M$, for some $M$
 and $\widetilde x$ with $ |\widetilde x| = n$.
The unsolvable of order $\omega$ can produce unboundedly many abstractions while performing head
reductions. 

\DSe

\begin{definition}[L\'{e}vy-Longo trees and B\"ohm trees]
\label{d:LTS}\index{Levy--Longo@\LLT}
The {\em L\'{e}vy--Longo Tree} of $M \in \Lambda$ 
is the labelled tree, LT(M), defined coinductively as follows:
\begin{enumerate}\renewcommand{\labelenumi}{(\theenumi)}
\item
 $LT(M) = \top$  if
 $M$ is an unsolvable of order $\infty$;

\item 
 $LT(M) = \lambda x_1 \ldots x_n. \bot $ 
 if $M$ is an unsolvable of order
$n$;


\item 
 $LT(M) =$ tree with $\lambda \ve{x}.y$ as the root and $LT(M_1)$,...,$LT(M_n)$ as the children, if $M$  has 
head normal form 
$\lambda \widetilde{x}.yM_1 \ldots  M_n$, $n \geqNAT 0$. That is, \\
 $LT(M) = $  \\
\begin{tikzpicture}[level distance=10mm,sibling distance=5mm]
  \node {$\lambda \ve{x}.y$} [grow=down]
  child   {node{$LT(M_1)$}} 
  child[missing]   {node{2}} 
  child[missing]   {node{3}} 
  child { node {$\cdots$} edge from parent[draw=none]  node {$\cdots$}} 
  child[missing] {node{5}}
  child[missing] { node {6} }
  child {node{$LT(M_n)$}};
\end{tikzpicture}
\end{enumerate}

 Two terms $M,N$ have the same  LT if  $LT(M) = LT(N)$.  
The definition of B\"ohm trees (BTs)
 is obtained from that of LTs using BT  in place of LT  in the
 definition above, and demanding that 
$BT(M) = \bot$ whenever  $M$  is unsolvable. 
That is, \Db clauses (1) and (2)  are replaced by \De the following one:
\[
BT(M) = \bot \qquad \mbox{ if  $M$  is  unsolvable}
\]

\end{definition}



\subsection{The (asynchronous) $\pi$-calculus}

We first consider encodings into the \emph{asynchronous}
$\pi$-calculus because its theory is simpler
 than that of the
synchronous $\pi$-calculus (notably bisimulation does not require
closure under name instantiations and has sharper congruence
properties \cite{BoSa98tcs}) 
and because it is the usual target
language for encodings  of the
 $\lambda$-calculus.
In all  encodings
we consider, 
  the encoding of a  $\lambda$-term
 is parametric on  a name, 
that is, is   a 
 function from names  to $\pi$-calculus  
 processes. We call such expressions {\em abstractions}.
  For the purposes of this paper
 unary abstractions, i.e.,  with only one parameter, suffice.
The actual instantiation of the parameter of an abstraction $F$ is done
via the {\em application} construct $\app F a$.
 We use $P,Q$ for process, $F$ for abstractions.
Processes and abstractions  form the set  of  {\em $\pi$-agents} (or
simply \emph{agents}), ranged
over by $A$. 
Small letters 
 $a,b, \ldots, x,y, \ldots$  
range  over the infinite set of names.
\DSb
Substitutions, ranged over by $\sigma$, 
 act on names; for instance $\sub{\tilc}\tilb$ represents the 
 substituting in which the $i$-th component of  $\tilb$ is replaced by the $i$-th
 component of  $\tilc$. 
\DSe 
The grammar of the calculus is thus:
$$ \begin{array}{ccll}
A & := & P \midd F & \mbox{(agents)}\\
P & := & \nil    \midd    \inp a \tilb . P    \midd    \out a \tilb 
   \midd     
  P_1 |  P_2   \midd    \res a P   \midd  ! \inp a \tilb . P   
\midd \app F a& \mbox{(processes)}  \\
F & := & \abs a P & \mbox{(abstractions)}
   \end{array}
 $$

Since the calculus is polyadic, 
we assume a \emph{sorting system} \cite{Mil99}
   to avoid disagreements  in the arities of
the tuples of names 
carried by a given name.
We will not present the sorting system 
because it is not essential. 
The reader should
take for granted that all agents described  obey  a sorting. 
\begin{xxenv}
A \emph{context}  $\qct$ of $\pi$    is a $\pi$-agent in which some
subterms have been replaced by the hole $\holem$ or, if the context is
polyadic, with indexed holes $\holei 1, \ldots, \holei n$; 
then 
 $\ct A$ or $\ct {\til A}$
 is the agent resulting from replacing the holes with the terms $A$ or
 $\til A$.
 \XXab \Xab
\Db \clx{I don't understand the following sentence in the context. Is this a definition of an abstraction (or process) context means? What is the "initial expression"? }
If the initial $\pi$-agent  was an abstraction, we call the context an \emph{abstraction $\pi$-context}; otherwise it is a \emph{process $\pi$-context}.
A hole itself may stand for an abstraction or a process.
\De
\Xae 
\XXae
A context is \emph{guarded} if the holes in it only 
appear
\DSb 
underneath
\DSe some prefix (input or output) \cite{Mil89, SW01a}; for example context $\inp a \tilb . (P | [\cdot])$ is guarded whereas $\res a (P | [\cdot])$ is not. 
A name is \emph{fresh} if it does not occur in the objects under consideration. 
\end{xxenv}
\Db In a restriction $\res b P$, inputs $a(\til b).P $ or $! a(\til b).P$, and abstraction $(b)P$,
names $b $ and $\til b$  are binders with scope $P$. \De
\DSb
As for the $\lambda$-calculus, we assume that $\alpha $-convertible terms
are identified.
\DSe


\begin{figure}[tb]
\begin{center}
\begin{tabular}{rlrl}
{\trans{ inp}}:& $ \inp  a\tilb .     P \stm{\inp a \tilb}P$ 
& \trans{ rep}:& $
 !\inp a \tilb  .  P   \stm{ \inp a \tilb}  P | ! \inp a \tilb . P  $
\quad if $a \not \in \tilb$
 \\[\mysp]
{\trans{ out}}:& $ \out a\tilb     \stm{\out a\tilb} \nil $
&
{\trans{ par}}:& $\displaystyle{   P \stm\mu   P' \over   P | Q   \stm\mu
P'| Q } $ if $\bn \mu \cap \fn Q = \emptyset $   \\[\mysp]
 \multicolumn{4}{c}{
  {\trans{com}}: $ \;\;    \displaystyle{ P \stm{\inp  a\tilc }P'
\hskip .4cm   Q
\stm{\res {\til{d}}\out a\tilb}Q'  \over     P  
|  Q \stm{ \tau} \res{\til{d}}( P' \sub \tilb\tilc
|  Q' )}$ if  $\til{d} \cap \fn P = \emptyset $
} \\[\mysp]
 \multicolumn{4}{c}{
{\trans{ res}}: \; \; $\displaystyle{ P \stm{\mu}P' \over
 \res a     P   \stm{ \mu} \res a P'  } $ $ a$ does not appear in $\mu$
} \\[\mysp]
 \multicolumn{4}{c}{
{\trans{ open}}:\;\; $\displaystyle{ P \stm{\res{\til{d}} \out a\tilb }P' \over
 \res c     P   \stm{ \res{c,\til{d}} \out a\tilb  }  P'  } $ $c\in \tilb
-\til{d}, \;  a \neq  c$.
} \\[\mysp]
 \multicolumn{4}{c}{
 {\trans{ app}}:  $ \; \; \displaystyle{ P\sub b a \stm{\mu}P' \over
 \app F b   \stm{ \mu}  P'  } $ if  $F = \abs a P$}
\end{tabular}
 \end{center} 
 \caption{Operational semantics of the $\pi$-calculus}\label{f:opsem_pi}
\end{figure}

The operational semantics of the asynchronous polyadic $\pi$-calculus  is standard \cite{SW01a},
and given in \afig \ref{f:opsem_pi}. 
We write  $\fn P$ for the free
names of a process $P$, and $\bn \mu$ for the bound names of action $\mu$.

Transitions are of the form $ P \stm{\inp a \tilb}P'$ (an input, $\tilb$
are the bound names of the input prefix that has been fired), 
$P
\stm{\res {\til{d}}\out a\tilb}P'$ (an output, where $\til d \subseteq
\tilb$ are private names extruded in the output), and $P \stm\tau P'$
(an internal action). We use $\mu$ to range over the labels of
transitions.  
We write
$\Ar {}$ 
for the reflexive transitive closure of $\stm{\tau}$, and 
$\wt{\mu}$ for $\wt{}\stm{\mu}\wt{}$; then
$\Arcap \mu$ is $\wt{\mu}$ if $\mu$ is not $\tau$, and $\wt{}$
otherwise; finally 
$P \stm{\widehat{\mu}} P'$  holds if  $P \stm{\mu}P'$  or ($\mu=\tau$
and $P = P'$).
In 
\Db bisimilarity \De  or similar coinductive relations for the asynchronous
$\pi$-calculus, no name instantiation
is required in the input clause or elsewhere
\Db  because such relations are already closed under name substitutions.
\begin{definition}[bisimilarity]
\label{d:bisimulation}
\XXb
A symmetric relation $\R$ on $\pi$-processes is a
\Db
\emph{bisimulation}, if whenever $P \,\R\, Q$ and $P \stm\mu P'$, then $Q \Arcap\mu Q'$
for some $Q'$ and $P' \,\R\, Q'$. 

Processes $P$ and $Q$ are \emph{bisimilar}, written
$P\approx Q$, if $P \,\R\, Q$ for some bisimulation $\R$. 
\De\XXe
\end{definition}
 In a standard way, we can extend $\approx$ to abstractions:
\Db
 $F 
\approx G 
$ if $\app {F} b \approx \app {G} b $ for every $b$.
\De As usual, strong bisimilarity, written $\sim$, is obtained by replacing $\Arcap\mu$ with $\stm\mu$ in the definition of weak bisimilarity,


A key  preorder in our work will be  \emph{expansion}
\cite{A-KHe92,SW01a}; this is a refinement of  bisimulation that takes into account the
number of internal \Db actions. \De Intuitively, $Q$ expands $P$ if they are weakly bisimilar and moreover $Q$ has no fewer internal actions when simulating $P$.

\begin{definition}[expansion relation]
\label{d:expa}
A relation $\R$   on  
 $\pi$-processes is an 
 \emph{expansion relation} if  whenever $P \RR
Q$: 
\begin{enumerate}
\item
 if  $P \stm\mu P'$ then $Q \Ar\mu Q'$ and $P' \RR Q'$; 
\item  if  $Q \stm\mu Q'$ then $P \arcap\mu P'$ and $P' \RR Q'$.
\end{enumerate}
\Db
We write $\expa$ for the largest expansion relation, and  call it
 \emph{expansion}.  
\De
 \end{definition}

We also need 
\Db
the  `divergence-sensitive' variant of expansion, \De 
written $\expa^\Div$,
as
an auxiliary relation when tackling must equivalences. 
Using $\Div$ to indicate  divergence (i.e., $P\Div$  if $P$ can
undergo  an
infinite sequence  of $\tau$ transitions), then  
$\expa^\Div$ is
obtained by adding into Definition~\ref{d:expa}  the requirement that
$Q\Div$ implies $P\Div$. 
We write 
$\contr$ and $\mbox{}^\Div{\contr}$ for the inverse of 
$\expa$ and 
$\expa^\Div$, respectively.
As instance of a contextual divergence-sensitive equivalence, we consider
\emph{must-termination},  because of the simplicity of its definition~---
other choices would have been possible.
\Db 
The predicate
$\Dwa$  indicates   barb-observability, i.e.,  $P\Dwa$ if $P
\Longrightarrow \stm \mu$
for some $\mu$ other than $\tau$. 
\De

\begin{definition}[may and must equivalences] 
The $\pi$-processes $P$ and $Q$ are \emph{may equivalent}, written 
 $P \may Q$,  if in  all process contexts $\qct$ we have  $\ct P \Dwa$ iff
 $\ct Q \Dwa$.
They are \emph{must-termination equivalent} (briefly \emph{must equivalent}), written 
 $P \must Q$,  
  if in   all process contexts $\qct$ we have  $\ct P \Div$ iff
 $\ct Q \Div$.
\end{definition} 

The behavioural relations defined above use the standard
observables of  $\pi$-calculus; they can be made coarser
by using the observables of asynchronous calculi, where one takes into
account that, since outputs are not blocking, only output transitions
from tested processes are immediately detected by an observer. 
In our   examples, the option of
asynchronous observable will make  a difference
 only in the case of may equivalence. 
In \emph{asynchronous may equivalence}, 
 $\mayasy$, 
the 
barb-observability predicate $\Dwa$ is replaced by 
the  asynchronous barb-observability
predicate $\outcon$, whereby $P \outcon$ holds if
$P\wt{}\stm{\mu}$ and $\mu$ is an output action. 
We have ${\expa} \subseteq {\approx}\subseteq{ \may}\subseteq {\mayasy}$, and 
${\expa^\Div }\subseteq {\must}$. 
The following  results will be useful later. 
\DSb
A process is
\emph{inactive} if it may never perform a visible action, i.e.,
 an input or output; formally $P$ is inactive if there is no $P'$ such
 that   
$P\wt{}\stm{\mu} P'$ and $\mu$ is   an input or output.  
\DSe

\begin{lemma}
\label{l:ina-div-NEW}
For all process  contexts $\qct$, we have: 
\begin{enumerate}
\item
 if $P$ is inactive, then 
 \begin{itemize}
 \item 
 $\ct P \Dwa$ implies $\ct Q \Dwa$ for all $Q$, 
\item $\ct P \outcon$ implies $\ct Q \outcon$ for all $Q$, 
\item   $\ct{\inp  a \tilx . P} \outcon$ implies $\ct{ P} \outcon$;
 \end{itemize}
\item  if $P\Div$ then  for all $Q$, $\ct Q \Div$ implies $\ct P \Div$.
\end{enumerate} 
%
\end{lemma}

\begin{lemma}
\label{l:comm}
$\res a( \out a {\til b } | \inp a {\til x}. P ) \contrDiv
P \sub{\til b}{\til x}$.
\end{lemma}


%
\section{Encodings of the $\lambda$-calculus and full abstraction}
\label{sec-encoding-notion}


\DSb
To make the encodings more readable, we shall 
 assume that $\lambda$-variables are included in the set of
$\pi$-calculus names. 
In this paper, an `encoding of the $\lambda$-calculus into
$\pi$-calculus'
 is supposed to be  {\em compositional}
and \emph{uniform}. 
Compositionality means that the 
definition of the encoding on a term  should  depend   only  upon the 
definition on the term's immediate
constituents, following the grammar of the encoded language. 
In the specific case of an encoding
$\qenco{}$
 of $\lambda$ into $\pi$, this means
that the encoding is defined thus:   
\[
\begin{array}{rcl}
\encom{x} &\defin& T_x \\
\encom{\lambda x.M} &\defin& \qctl[\encom{M}] \\ 
\encom{MN} &\defin& \qctapp[\encom{M},\encom{N}] 
\end{array}
\] 
where $T_x$ is a $\pi$-term that may contain 
 free occurrence of $x$, $\qctl$ is a $\pi$-context in which 
 $x$  
\DSb
only appears
\DSe  as  a bound name with a  scope that embraces the
 hole,  and $\qctapp$ is a two-hole $\pi$-context. 
\DSb
\DSe

Uniformity   refers to the treatment of the free variables: if the
$\lambda$-term $M$ and $M'$  are the same modulo a renaming of free variables, then also
their encodings should be same modulo a renaming of the corresponding free names.
A way of ensuring this is to require 
that the encoding commutes with  name substitution; i.e.,
if $\sigma $ is a $\lambda$-calculus variable renaming (a
 substitution  from variables to variables) that, 
since 
 $\lambda$-variables are included in the set of
 $\pi$-calculus names, also represents a   $\pi$-calculus substitution 
from names to   names, then 
  it holds that
\Db
 \begin{equation}
\label{e:MNsigma}
\encom {M \sigma} \equiv \encom M \sigma \, . 
\end{equation}
\DSe
This condition, involving substitutions, is meaningful provided that 
the encoding respects $\alpha $-conversion; that is, if $M $ and $N$ are
$\alpha$-convertible terms, then $\encom M = \encom N$.  
Moreover, if  we take $\sigma$ to mean a substitution acting on 
the set of  $\pi$-calculus names (a superset of  the set of $\lambda$-calculus variables),
then  \reff{e:MNsigma} also says that the encoding does not introduce extra free names;
that is,  for any $M$, the free names of $\encom M $ are also
 free variables of $M$.   Thus uniformity comes with
condition \reff{e:MNsigma}, plus the 
these  conditions on
$\alpha$-conversion and on free names.

\De

\DSb
A compositional encoding can be extended  to contexts, by extending 
the encoding mapping  a
$\lambda$-calculus  context into the corresponding 
 $\pi$-calculus context. 
Two such contexts   will be useful in this
work, for a given encoding $\qenco{ }$:
\begin{enumerate}[label=(\arabic*)]
\item $ \qctl \defin \encom{\lambda x . \holem}$,  called an
 {\em  abstraction  context of $\qenco$}. 
We have already mentioned 
this encoding when describing the meaning of 
 compositionality.
\item $ \qctappn \defin \encom{x\holem_1\cdots\holem_n}$ (for $
  n\geqNAT 0$), called 
\DSb   a   {\em  variable  context of $\qenco$}.
This context  will be used to represent the encoding of terms of the
form $x M_1 \cdots M_n$, for some $M_1\cdots M_n$, as we have 
\[ 
\encom {x M_1 \cdots M_n}  \equiv
\ctappn {\encom { M_1 }, \cdots, \encom { M_n }}
\]
\end{enumerate}
\DSe


In the remainder of the paper, `encoding' refers to a `compositional
and uniform encoding of
the $\lambda$-calculus into the $\pi$-calculus'.

\begin{definition}[soundness, completeness, full abstraction, validity
  of $\beta$ rule]
An encoding $\qenco$   and a  relation $\R$ on $\pi$-agents  are:  
 \begin{enumerate}
 \item  {\em sound for LTs}
if $\encom M \RR \encom N$ implies 
$\LLT M = \LLT N$,  for all $M, N \in \Lambda$; 
\item  
 {\em complete for LTs}
if  $\LLT M = \LLT N$
implies $\encom M \RR \encom N$,  for all $M, N \in \Lambda$;
\item  
 \emph{fully abstract for LTs} if they are both sound and
  complete for LTs. 
\end{enumerate}
The same definitions are also applied to BTs~--- just replace
`LT' with `BT'.
Moreover,  $\qenco$ and $\R$ 
   {\em validate rule   $\beta$}  if 
$\encom{(\lambda x . M)N} \rela \encom{M\sub N x}$, for all $x,M,N$.
\end{definition}


\section{Conditions for completeness and soundness}
\label{ss:sounds_completes}


We first  give the  conditions for {completeness}
of an
encoding
$\qenco$
 from the $\lambda$-calculus into $\pi$ with respect to  a relation
$\asymp$ on $\pi$-agents; then those for {soundness}.
In both cases, the conditions  involve an auxiliary relation
$\leq$ on $\pi$-agents.
\Db
\De

\subsection{Completeness conditions.} 
In the conditions for  completeness
the auxiliary precongruence $\leq$ 
is required 
\Db 
\De
to validate an `up-to $\leq$ and contexts'
 technique. Such technique is inspired by  
the  `up-to expansion and contexts' technique for  bisimulation   
\cite{SW01a}, which   allows us the following flexibility 
in the  bisimulation game required on a candidate  relation $\R$:
given a  pair of derivatives $P$ and $Q$,  it is not necessary
that the pair $(P,Q)$ itself  be in $\R$, as in the ordinary definition
of bisimulation; it is sufficient to find processes $\til P,\til
Q$, and a
context $\qct$ such  that $P \contr  \ct {\til P}$,  $Q \contr
\ct{\til Q}$, and $\til P \RR \til Q$; that is, we can manipulate the
 original derivatives in terms of $\expa$ so to isolate a common context
 $\qct$; this context is removed and  only the resulting processes
 $\til P,\til Q$ need to be  in $\R$.
In the technique,  the expansion relation is important: replacing it
with   bisimilarity  breaks correctness.
Also, some care is necessary when a hole of the  contexts occurs
underneath an input prefix, in which  case a closure under name
substitutions is required. 
Below, the technique is formulated in an abstract manner, using generic
relations 
  $\asymp$ and $\leq$.
   In the encodings we shall examine,   $\asymp$ will be any
of the congruence relations in Section~\ref{sec-background}, whereas 
 $\leq$ will always be the expansion relation (or its
   divergence-sensitive variant, when $\asymp$ is  must
   equivalence).


\begin{definition}[\uptoc\ technique]
\label{d:uniCAsyn}
\

\DSb
\begin{itemize}
\item
A  symmetric relation  $\R$ on $\pi$-processes
is an \emph{\uptoc\ candidate for $\asymp$}  if
for any pair $(P,Q)\in \R $, if  $P \stm \mu P'$ then  $Q \Arcap \mu
Q'$ and there  are processes $ \til P,  \til Q$ and 
a  context  $\qct$  such  that 
  $P' \geq  \ct {\til P} $, $Q' \geq  \ct {\til Q}$, 
and, if $n\geqNAT 0 $ is the length of the tuples $\til P$ and $\til Q$, 
 at least one of the following two statements is true, for each 
\Db\Xb
$1\leqNAT i \leqNAT n$: \Xe
\De\DSb
\begin{enumerate}
\item 
 $P_i \asymp Q_i$; 
\item
 $P_i \RR Q_i$ and, 
 if $[\cdot]_i$ occurs underneath an input-prefix in $C$ (that is, guarded by an input),
also $P_i\sigma \RR  Q_i\sigma$  for all substitutions  $\sigma$.
\end{enumerate}
\DSe
\item 
Relation  $\asymp$ {\em validates the \uptoc\ technique} if 
for any 
 \uptoc\ candidate  for $\asymp$  $\R$
we have $\R \subseteq {\asymp}$.
\end{itemize}
 \end{definition}

 Below is the core of the completeness conditions (Definition \defref{d:faith}). Some of these
 conditions
\Db
 (\hardcode{(1)}\conref{i:sLLa} to \hardcode{(3)}\conref{i:sLLf})  
\De
only concern the chosen behavioural equivalence
 $\asymp$ and its auxiliary relation $\leq$, 
and  are independent of the encoding; 
the most important condition is the validity
of the \uptoc\ technique. 
\DSb
Other conditions (such as \hardcode{(4)}\conref{i:sLLb}) are purely 
syntactic;  we use the standard concept of \emph{guarded context}, 
in which each hole appears underneath some prefix \cite{Mil89, SW01a}. 
\DSe 
The only behavioural
 conditions on the encoding 
\DSb
are \hardcode{(5)}\conref{i:sLLg} and 
\hardcode{(6)}\conref{i:sLLh} in Definition \defref{d:faith}, plus
(ii) in Theorem~\ref{t:compNEW}. They \DSe
 concern validity of
 $\beta$ rule and equality of certain unsolvables~---
 very basic requirements for the operational  correctness of an
 encoding.

\finish{below: major changes} 
\DSb
We recall that a relation $\R$  in a language
 that is preserved by the
constructs of the language is:
\Db
\begin{itemize}
\item 
 a  \emph{precongruence} if $\R$  is a
preorder relation;  

\item 
a   \emph{congruence} if $\R$ is an equivalence
relation.  
\end{itemize}
\De

Note that for any abstraction $F \defin (a)P$, the terms $
\app F b$ and $P \sub b a$ have exactly the same transitions. 
Hence we expect any behavioural relation $\R$ to identify such
processes, i.e., $\app F b \RR P \sub b a$ as well as 
$P \sub b a \RR \app F b $. 
We call \emph{plain} a relation on processes in which this holds.

\begin{lemma}
\label{l:plain}
\Db 
Let $\R$ be a  plain   precongruence
 on $\pi$-agents. We have: 
\De
\begin{enumerate}
\item
 $\R$ is preserved by name
substitutions, i.e., $P \RR Q$ implies 
$ P \sub ab  \RR Q \sub ab$, for all $a,b$;

\item 
if $F\defin (a)P$ and 
$G\defin (a)Q$ 
then  
$F \RR G$ implies $P \RR Q$, and 
 $\app F z \RR \app G z $ implies
$F \RR G$, for any fresh name $z$.
\end{enumerate}
 
\end{lemma}

\begin{proof}
For the first item, 
 from $P \RR Q$, by the precongruence property we have 
$(a) P \RR (a)Q$ and then also 
$\app{((a) P)} b  \RR \app{((a)Q)} b $; since $R$ is plain
we  conclude $P \sub ba \RR  Q \sub ba$.

For the second item, in the first case from 
$F \RR G$ we derive 
 $\app F a \RR \app G a $, hence 
also $P \sub aa \RR Q\sub a a $, which is 
 $P  \RR Q $. The second case is similar:
$\app F z \RR \app G z $ implies 
$(z)(\app F z) \RR (z)(\app G z) $ that,
 using the plain and precongruence properties, implies
$(z)(P\sub z a ) \RR (z)(Q\sub z a) $, which is the same as 
$F  \RR G $, since $z$ is fresh and we identify
$\alpha$-convertible terms. 
\end{proof} 
\DSe
 
\begin{definition}
\label{d:faith}
Let 
$\asymp$ and $\leq$ be relations on 
$\pi$-agents such that: 

\begin{enumerate}\renewcommand{\labelenumi}{(\theenumi)}
\item \label{i:sLLa} 
\DSb
$\asymp$ is a congruence, and 
  $  {\asymp} \supseteq{ \geq}$;
\DSe
\item 
\label{i:expa} 
$\leq$ is an expansion relation and is a 
\DSb
plain precongruence;
\DSe

\item $\asymp$ validates the \uptoc\ technique. 
\label{i:sLLf} 
\end{enumerate} 
Now,  an encoding $\qenco$  of  $\lambda$-calculus into
$\pi$-calculus is \emph{faithful for $ 
\asymp$ under  $\leq$}
if 
\begin{enumerate}\renewcommand{\labelenumi}{(\theenumi)}\setcounter{enumi}{3} 
\item the  variable  contexts of $\qenco$ are guarded;
\label{i:sLLb} 


\item 
\label{i:beta}
 $\qenco$ and $\geq$
validate rule $ \beta$;
\label{i:sLLg} 


\item if   $M $ is an unsolvable of order $0$ then $\encom M \asymp \encom \Omega$. 
\label{i:sLLh} 
\end{enumerate}
\end{definition}

\begin{theorem}[completeness]
\label{t:compNEW}
Let $\qenco$ be an encoding of the $\lambda$-calculus into
$\pi$-calculus, and  $\asymp$ a relation on $\pi$-agents. 
Suppose  there is a 
relation  $\leq$ on $\pi$-agents
such that $\qenco$ is faithful for $ 
\asymp$ under  $\leq$.  
We have:
\begin{enumerate}[label=(\roman*)]
\item \label{i:sLLj} 
if  the abstraction contexts of $\qenco$ are guarded, 
then
 $\qenco$ and $\asymp$ are complete  for LTs;

\item \label{i:sLLk} 
if
$ \encom {M} \asymp \encom{\Omega} $ whenever $M$ is unsolvable
 of order $\infty$,
 then 
 $\qenco$ and $\asymp$ are complete  for BTs.
\end{enumerate}
 \end{theorem}

\Db
\noindent The proof of Theorem \ref{t:compNEW} is placed in Appendix \ref{ap:proof_conditions}. 
We provide some intuitive account below. 
The proofs for LTs and BTs are similar.
In the proof for LTs, for
instance,  we consider the relation 
\[\R \defin \{(\app {\encom M} r , \app {\encom N} r)  \,|\,
\begin{array}[t]{l}
\LLT M = \LLT N ,
\mbox{ and $r$  fresh}
\}    
\end{array}
 \]
and   show that for each $(\encom M,\encom N)\in \R$ one of the following conditions is
true, for some abstraction context $\qctl$, variable context 
 $\qctappn$, and terms $M_i,N_i$:

\begin{enumerate}[label=(\alph*)]
\item $\encom M \asymp \encom \Omega $ and $\encom N \asymp \encom
\Omega$;

\item $\encom M \geq \ctl {\encom{M_1}}$,
$\encom N \geq \ctl {\encom{N_1}}$ and
$(\encom{M_1},\encom{N_1})\in \R$; 

\item $\encom M \geq \ctappn {\encom{M_1}, \ldots , \encom{M_n} }$,
$\encom N \geq \ctappn {\encom{N_1} , \ldots , \encom{N_n}   }$  and  
$(\encom{M_i},\encom{N_i})\in \R$  for all $i$.
\end{enumerate} 
Here, (a) is used when $M$ and $N$ are unsolvable of order $0$, by
 appealing to clause  \reff{i:sLLh} of Definition~\ref{d:faith}. 
In the remaining cases we obtain (b) or (c), depending on the shape of
the LT for $M$ and $N$, and appealing to 
 clause  \reff{i:beta} of Definition~\ref{d:faith}. 
The crux of the proof is exploiting  
the property that 
$\asymp$ validates the \uptoc\ technique
so to  derive ${\R} \subseteq {\asymp}$ (i.e., 
the continuations of $\encom M$ and $\encom N$ are related, using expansion and cutting
a common context).
Intuitively, this is possible because
  the 
variable and abstraction contexts of $\qenco$ are guarded 
and because $\leq$ is an expansion relation (clause \reff{i:expa}
 of Definition~\ref{d:faith}).
%
In the results for BTs, the condition on abstraction contexts being
guarded is not needed because the condition can be proved redundant in
presence of the  condition in the assertion  (ii) of the
theorem. 

\De



\subsection{Soundness conditions}
In the conditions for  soundness, one of the key requirements will be
that certain contexts have an \emph{inverse}.
This intuitively means that
 it is possible to  extract any of the processes in the holes of the
 context, up to the chosen behavioural equivalence. 
To have some more flexibility, we allow the appearance of the
process of a hole after a rendez-vous with the external observer. 
This allows us to:  initially  restrict some  names that are used to
consume the context;   then  export such names before revealing the
process of the hole. The reason why the restriction followed by  the
export of these names is useful is that the
 names might occur in the process of the hole; initially restricting
 them allows us to hide the names to the external environment;  
 exporting them allows to remove the restrictions
once the inversion  work on the context  is completed. 
The drawback of this initial rendez-vous is that we have to require
a prefix-cancellation property on the behavioural equivalence;
however, the requirement is straightforward to check in
common behavioural equivalences.

\XXb
We give the definition of inversion 
only for abstraction $\pi$-contexts whose holes are
themselves abstractions; 
\DSb
that is,  
 contexts that
are obtained by a $\pi$-abstraction by   replacing subterms that are
themselves abstractions with holes. 
\XXe
We only
need this form of contexts when reasoning on $\lambda$-calculus
encodings, and each hole 
of a context
 will be filled with the  encoding of a $\lambda$-term.


\begin{definition}\label{d:inverse_conext}
Let $C$ be an abstraction $\pi$-context
   with $n$ holes,
each occurring exactly once, each hole itself standing for an
abstraction.  
We say that \emph{$C$ has inverse with respect to a relation $\R$} on $\pi$-agents, if for
every $i=1,...,n$  and for
every $\til A $
there exists a process $\pi$-context $D_i$ 
and fresh names
 $a,z,b $
such that 
$$D_i[C[\ve{A}]] \;\R\; (\res {\til b})
(\out a {\til c}| b(z). \app{A_i} z) \; , 
\hskip 1cm  
\mbox{for 
 $b \in  \til b \subseteq \til
c$.}
$$
\end{definition}

It is useful to
 establish inverse properties for contexts for
the finest possible  behavioural relation,  so to export  the
\Db
 results
 onto \De coarser relations. In our work, the finest
such relation is the divergence-sensitive expansion ($\expaDiv$).

\begin{example}
\label{exa:inverse}
We show examples of inversion using contexts that are similar to
some abstraction and variable contexts in encodings of
$\lambda$-calculus.
\begin{enumerate}
\item
Consider  a  context 
$ C \defin (p ) \:  p(x,q) . (\app \holem q)$.
If $F$ fills the context, then an inverse 
for $ \contrDiv$
is the context  
\[ D \defin    \res b (  \out a  b | b(r) . \res p ( \app {\holem} p  | \out p {x,r} )) \hskip 1cm
  \] 
where all names are \Db 
fresh. \De 
Indeed we have,
using simple algebraic manipulations (such as the law of
Lemma~\ref{l:comm}):
\[
\begin{array}{rcl}
  D [\ct F ] &  \contrDiv &     \res b (  \out a  b | b(r) .  \res p (
 p(x,q) . \app F q   | \out p {x,r} )) \\
 &  \contrDiv &    \res b (  \out a  b | b (r).  
\app F r) 
\end{array}
 \] 

\item Consider now
 a context 
$
 \qct \defin \abs p (\res{r,y}) (\out
{x}{r}  |  \out r {y, p} | !  \inp y q .  \app \holem  q )
$. 
If $F$ fills the hole, then an inverse context is 
\begin{equation}
\label{e:hole}
D \defin  ( (\res{ x,p,b })  (\app \holem p  | x(r).  r(y,z) . 
(\out a {x,b} | b (u).  \out y
u)    ) 
\end{equation} 
where again all names are fresh with respect to $F$. We have:
\[
\begin{array}{rcl}
  D [\ct F ] & = &    
 ( \res{ x,p,b })  (\app {(C [F])} p  | x(r).  r(y,z) . 
(\out a {x,b} | b (u).  \out y
u)    )  \\
 & \contrDiv &
 (\res{ x,p ,b})  ( \res{r,y}) (\out
{x}{r}  |  \out r {y, p} | !  \inp y q .  \app F  q ) | \\
 & & \qquad\qquad\qquad\quad\;\,  x(r).  r(y,z) . (\out a {x,b} | b (u).  \out y u)    \\
 & \contrDiv &
  \resb{ x,b }  ( \res{y} ( !  \inp y q .  \app F  q  | 
(\out a {x,b} | b (u).  \out y
u)    ) ) \\
 & \contrDiv &
  \resb{ x,b } (\out a {x,b} | b (r).  
 ( \res{y} ( !  \inp y q .  \app F  q  | 
 \out y r)  )) \\
 & \contrDiv &
  \resb{ x,b } (\out a {x,b} | b (r).     \app F  r  ) 
\end{array}
 \] 
\end{enumerate}
\end{example}


\begin{definition}
\label{d:rv_canc}
A relation $\R$ on $\pi$-agents \emph{has the rendez-vous
  cancellation} property  if whenever 
$ \res {\til b} ( \out a {\til c} | b(r). P ) \RR 
 \res {\til b} ( \out a {\til c} | b(r).  Q)$
where $b \in \til b \subseteq \til c $ and $a,b$ are fresh, 
then also $P \RR Q$.  
\end{definition} 

The   cancellation property is straightforward
for a behavioural relation $\asymp$
 because, in the initial processes,  the output $\out a {\til c}$ is the only possible  initial action,
 after which the input at $b$ must fire (the assumption `$a,b$ fresh'
 facilitates matters, though it is not essential).

As for completeness, so for soundness we isolate the common conditions
for LTs and BTs. 
Besides the conditions on inverse of contexts, the other main
requirement 
is about  the inequality 
among some  structurally different $\lambda$-terms
 (condition
\conref{con:lls:f}).

\begin{definition}
\label{d:resp}
Let 
$\asymp$ and $\leq$ be relations on 
$\pi$-agents where 
\begin{enumerate}\renewcommand{\labelenumi}{(\theenumi)}
\item \label{con:lls:a} 
\DSb
$\asymp$  is a congruence, $\leq$  a plain precongruence; 
\DSe
\item   $ { \asymp} \supseteq {\geq}$;  
\label{con:lls:b} 

\item \label{con:lls:c} 
  $ { \asymp}$ has the rendez-vous
  cancellation property.
\end{enumerate}
An encoding  $\qenco$ of the $\lambda$-calculus into
$\pi$-calculus is \emph{respectful for $ 
\asymp$ under  $\leq$}
if
\begin{enumerate}\renewcommand{\labelenumi}{(\theenumi)} \setcounter{enumi}{3}
\item 
$\qenco$ and $\geq$ validate rule $\beta$;
\label{con:lls:d}

\item \label{con:lls:e} 
if $M$ is an unsolvable of order $0$, then $\encodingm{M}{}{} \asymp \encodingm{\Omega}{}{} $;

\item \label{con:lls:f}
the terms 
$\encodingm{\Omega}{}{}$, $\encodingm{x\ve M}{}{}$, $\encodingm{x\ve {M'}}{}{}$,
and $\encodingm{y\ve {M''}}{}{}$
are 
pairwise unrelated by $\asymp$, assuming that  
$x \neq y$  and  that  tuples $\ve M $ and $\ve {M'} $
have different  lengths;

\item 
\label{i:inverse}
the abstraction 
and  variable contexts 
 of $\encodingm{\,}{}{}$
 have  inverse with respect to $\geq$.
\label{con:lls:gh} 

\end{enumerate}
\end{definition} 

The condition on variable context having an inverse is the
most delicate one. In the
  encodings of the $\pi$-calculus we have examined,  however, the condition  is
simple  to achieve.




\begin{theorem}[soundness]
\label{t:soundNEW}
Let $\qenco$ be an encoding of the $\lambda$-calculus into
$\pi$-calculus, and  $\asymp$ a relation on $\pi$-agents. 
Suppose  there is a 
relation  $\leq$ on $\pi$-agents
such that $\qenco$ 
 is respectful for $ 
\asymp$ under  $\leq$.
We have:
\begin{enumerate}[label=(\roman*)]
\item \label{con:lls:j}
if, for any $M$, the term
$\encodingm{\lambda x.M}{}{}$
 is unrelated by  
$\asymp$  to 
$\encodingm{\Omega}{}{}$ 
and to 
any term  of the form 
$\encodingm{x \ve M}{}{}$,
then
 $\qenco$ and $\asymp$ are sound  for LTs;

\item \label{con:lls:k}
if 
\begin{enumerate}
\item
$ \encom {M} \asymp \encom{\Omega} $ whenever $M$ is unsolvable
 of order $\infty$,
\label{con:bts:e}

\item \label{con:bts:f} 
 $M$  solvable implies that the term
$\encodingm{\lambda x.M}{}{}$ is unrelated by  
$\asymp$  to 
$\encodingm{\Omega}{}{}$ 
and to 
any term  of the form 
$\encodingm{x \ve M}{}{}$,
\end{enumerate}
 then 
 $\qenco$ and $\asymp$ are  sound  for BTs.
\end{enumerate}
 \end{theorem} 

For the proof of Theorem~\ref{t:soundNEW}, we use a coinductive
definition of LT and BT equality, as forms of bisimulation. 
Then we show that the relation
 $\{(M,N) \,|\; \encodingm{M}{}{} \asymp \encodingm{N}{}{} \}$
implies the corresponding tree equality. In the case of internal
 nodes 
of the trees, we exploit conditions such as \conref{con:lls:f} and
\conref{i:inverse} of Definition~\ref{d:resp}. 
The details are given in Appendix \ref{ap:proof_conditions}.



\paragraph{Full abstraction}
We put together
Theorems~\ref{t:compNEW}
and \ref{t:soundNEW}.

\begin{theorem}
\label{t:fa}
Let $\qenco$ be an encoding of the $\lambda$-calculus into
$\pi$-calculus,   $\asymp$ a congruence on $\pi$-agents.
Suppose  there is a 
\DSb
plain
\DSe
  precongruence
  $\leq$ on $\pi$-agents
such  that

\begin{enumerate}\renewcommand{\labelenumi}{(\theenumi)}
\item
$\leq$ is an expansion relation and   ${  \asymp} \supseteq{ \geq}$;
\item $\asymp$ validates the \uptoc\ technique;

\item the  variable contexts  of $\qenco$ are guarded;

\item 
the abstraction
and   variable contexts
 of $\encodingm{\,}{}{}$
 have  inverse with respect to $\geq$;

\item
 $\qenco$ and $\geq$
validate rule $ \beta$;
\item if   $M $ is an unsolvable of order $0$ then $\encom M \asymp \encom \Omega$; 

\item
the terms
$\encodingm{\Omega}{}{}$,  
$\encodingm{x\ve M}{}{}$, $\encodingm{x\ve {M'}}{}{}$,
and $\encodingm{y\ve {M''}}{}{}$
are
pairwise unrelated by $\asymp$, assuming that
$x \neq y$  and  that  tuples $\ve M $ and $\ve {M'} $
have different  lengths.
\end{enumerate}
We have:
\begin{enumerate}[label=(\roman*)]
\item
if 
\begin{enumerate}

\item
  the   abstraction contexts  of $\qenco$ are guarded, and

\item
for any  $M$  the term
$\encodingm{\lambda x.M}{}{}$ is unrelated by
$\asymp$  to
$\encodingm{\Omega}{}{}$
and to
any term  of the form
$\encodingm{x \ve M}{}{}$,
\end{enumerate}
  then
 $\qenco$ and $\asymp$ are fully abstract  for LTs;

\item
if
\begin{enumerate}
 \item
 $M$  solvable implies that the term
$\encodingm{\lambda x.M}{}{}$ is unrelated by
$\asymp$  to
$\encodingm{\Omega}{}{}$
and to
any term  of the form
$\encodingm{x \ve M}{}{}$, and
\item
$ \encom {M} \asymp \encom{\Omega} $ whenever $M$ is unsolvable
 of order $\infty$,
\end{enumerate}
     then
 $\qenco$ and $\asymp$ are fully abstract  for BTs.
\end{enumerate}
 \end{theorem}

 In   Theorems~\ref{t:compNEW}(i)
and \ref{t:fa}(i)
for LTs
the abstraction contexts are required to be guarded.
This is reasonable in encodings of strategies, such as
 call-by-name, where
   evaluation does not continue
underneath a $\lambda$-abstraction, but it is too demanding
when evaluation can go past a $\lambda$-abstraction, such as
 strong call-by-name.
We therefore present
also
the following  alternative   condition:
\[
\mbox{$ $} \hfill \hskip 2cm
\mbox{$M,N $ unsolvable of order $\infty$
implies $\encom M \asymp  \encom N$.  }
\hfill \hskip 2cm  \mbox{$(*)$}
\]

\begin{theorem}
\label{t:sound-strong}
  Theorems~\ref{t:compNEW}(i)
and \ref{t:fa}(i)  continue to hold when
the condition
that the abstraction contexts be guarded
is replaced by $(*)$
above.
\end{theorem}

\DSb
The proof of Theorem~\ref{t:sound-strong} can be found in Appendix~\ref{ap:proof_conditions}.  
\DSe


%
\section{Examples with call-by-name }
\label{sec-examples-cbn}

In this section we apply the theorems on soundness and completeness in
the previous section to two well-known encodings of call-by-name
$\lambda$-calculus:
the one in \afig \ref{f:example-lazy}.a is Milner's original encoding
\cite{Mil92}.
The one  in \afig \ref{f:example-lazy}.b is a variant encoding in which
a function communicates with its environment via a rendez-vous
(request/answer) pattern. An advantage of this encoding is that it can
be easily tuned to call-by-need, or even used in combination with
call-by-value \cite{SW01a}.

For each encoding we consider soundness and completeness with respect
to four behavioural equivalences: bisimilarity ($\approx$), may
($\may$), must ($\must$), and
asynchronous may ($\mayasy$).
The following lemma allows us to apply the \uptoc\ technique.

\begin{lemma}
\label{l:vali}
Relations $\approx$, $\may$, and $\mayasy$ validate the \uptoE\ technique;
relation
 $\must$ validates the \uptoEdiv\ technique.
\end{lemma}

The
result in Lemma~\ref{l:vali}
 for the bisimulation is from \cite{SW01a}. The proofs for the may
equivalences follow the definitions of the equivalences, reasoning by
  induction on the number of steps required to
bring out an observable. The proof for the must equivalence uses
coinduction to reason on  divergent paths. Both for the may  and for
the
 must equivalences,  the
role of expansion ($\expa$) is similar to its role in the technique for
bisimulation.
Detailed discussion can be found in Appendix \ref{ap:property_eg}.

\begin{theorem}
\label{t:resCBN}
The encoding of Figure~\ref{f:example-lazy}.a
is fully abstract for LTs when the behavioural equivalence for
$\pi$-calculus   is $\approx, \may$, or $\must$; and fully abstract
for BTs when the behavioural equivalence is $\mayasy$.

The encoding of Figure~\ref{f:example-lazy}.b
is fully abstract for LTs under any of the equivalences
 $\approx, \may$, $\must$, or $\mayasy$.
\end{theorem}
As Lemma~\ref{l:vali} brings up,
 in the proofs,
the auxiliary relation
  for  $\approx$, $\may$, and $\mayasy$
 is $\expa$; for
 $\must$ it is $\expa^\Div$.
With   Lemma~\ref{l:vali} at hand, the proofs for the soundness and
completeness statements  are simple. Moreover, there is a large reuse
of proofs and   results.  For instance, in the  completeness results for LTs, we only have to check
that: the variable and abstraction contexts of the encoding are
guarded;  $\beta$ rule is validated;   all unsolvable of order $0$ are
equated. The first check  is
straightforward and is done only once.
For the $\beta$ rule, it
 suffices to establish its
validity  for $\expa^\Div$, which  is the finest among
the behavioural relations considered; this is done using
distributivity laws for private replications \cite{SW01a}, which
are valid for strong bisimilarity and hence for $\expa^\Div$, and
the law of Lemma~\ref{l:comm}. Similarly,
for the  unsolvable terms of order $0$
it suffices to prove
that they are  all `purely divergent', i.e., divergent and unable
to even perform some visible action, and this follows from
the validity of the $\beta$ rule for $\expa^\Div$.

Having checked the conditions for completeness, the only two
additional conditions needed for
 soundness for LTs are conditions \conref{con:lls:f} and \conref{i:inverse} of
 Definition~\ref{d:resp}, where we have to prove that certain terms
 are unrelated and that certain contexts have an inverse. The
 non-equivalence of the terms in  condition \conref{con:lls:f} can
be established for the coarsest equivalences, namely $\mayasy$
 and $\must$, and then exported to the other equivalences. It suffices to look at visible traces of length $1$ at
 most, except for terms of the form
 $\encodingm{x\ve M}{}{}$ and $\encodingm{x\ve {M'}}{}{}$,
when   tuples $\ve M $ and $\ve {M'} $
have different  lengths,
in which case one reasons by induction on the shortest of the two tuples (this argument is straightforward on the basis of the `inverse context' property explained below).

The most delicate point is
the `inverse context' property, i.e., the existence of an inverse for the
 abstraction
and the  variable contexts. This can be established for the finest
equivalence ($\expaDiv$), and
then exported to coarser equivalences.
The two constructions needed for this are similar to  those examined
in Example~\ref{exa:inverse}.
We give detailed proofs concerning `inverse context' for the examples in Appendix~\ref{ap:inverse_context_eg}.

\begin{figure}[t] 
{\small
\centering
 \begin{minipage}{0.45\textwidth}%
\centering
\[
\begin{array}{rcl}
\encom{ \lambda     x . M } & \defin &  \abs p \inp{ p}{x, q}.  \enco{M}{q} \\[3pt]
\encom{x} & \defin &    \abs p \out{x}{p}  \\[3pt]   
\encom{M N } & \defin &   \abs p (\res{r, x}) \Big(\enco{M}{r}  |  \out r {x, p} | \\[3pt]
& &   !  \inp x q .  \enco N q \Big) \quad\mbox{(for  $x$ fresh)}
\end{array} 
\]\vspace*{-.5cm}
\caption*{Figure \ref{f:example-lazy}.a: Milner's
encoding}
\label{robin-original-lazy} 
\end{minipage}%
\hspace{19pt}
 \begin{minipage}{0.45\textwidth}%
\centering
\[
\begin{array}{rcl}
\encom{ \lambda     x . M } & \defin &  \abs p \res v (\out{ p}{v} \para \inp {v} {x,q} . \enco{M}{q}) \\[3pt]
\encom{x} & \defin &    \abs p \out{x}{p}  \\[3pt]   
\encom{M N } & \defin &   \abs p \res{r} \Big(\enco{M}{r}  |  \\[3pt] 
&  & \inp r v. \res{x} (\out v {x, p} | \\[3pt]
&  &   !  \inp x q .  \enco N q)\Big)  \quad\mbox{(for  $x$ fresh)} 
\end{array} 
\]\vspace*{-.5cm}
\caption*{Figure \ref{f:example-lazy}.b: a variant
encoding 
}\label{robin-var-lazy-lpi} 
\end{minipage}%
}
 \caption{The two encodings of  call-by-name} \vspace*{-.4cm}
\label{f:example-lazy}
 \end{figure}



For Milner's encoding,
in the case of $\mayasy$,
we actually obtain   the BT equality. One
may find this surprising at first: BTs are defined from  weak
head reduction, in which evaluation continues underneath a
$\lambda$-abstraction; however Milner's encoding mimics the
call-by-name strategy, where reduction stops  when a
$\lambda$-abstraction is uncovered.
We obtain BTs with $\mayasy$ by exploiting
Lemma~\ref{l:ina-div-NEW}(1) as follows. The encoding of a term $\lambda x. M$ is
$  \abs p \inp{ p}{x, q}.  \enco{M}{q} $.
In an asynchronous semantics, an input is not directly observable;
with $\mayasy$
an   input prefix can actually be erased provided,  intuitively,
that an  output is never  liberated.
We sketch the  proof of
$ \encom {M} \mayasy \encom{\Omega} $ whenever $M$ is unsolvable
 of order $\infty$, as required in
 condition (ii) of Theorem~\ref{t:fa}.
Consider a context $\qct$ with $\ct {\encom M}
 \Dwa$, and  suppose the observable is reached after $n$ internal
 reductions. Term  $M$, as $\infty$-unsolvable, can be
 $\beta$-reduced    to $M' \defin  
 \lambda x_1...\lambda x_n.N$,
 for some $N$.  By validity of $\beta$-rule for $\contr$, also
$\ct{\encom{M'}}\Dwa$ in at most $n $ steps;
hence the subterm $\encom N $ of $\encom{M'}$ does not
contribute  to the observable,
since the abstraction
contexts of the encodings are guarded and $M'$ has $n$ initial
abstractions. We thus derive
$\ct{\encom{\lambda x_1...\lambda x_n.\Omega}}\Dwa$ and then,
by repeatedly applying the third statement of
Lemma~\ref{l:ina-div-NEW}(1) (as $\Omega$ is inactive),
 also $\ct{\encom{\Omega}}\Dwa$. 
The converse implication is given by the first statement in  Lemma~\ref{l:ina-div-NEW}(1).



\section{An example with  strong call-by-name }\label{sec-examples-scbn}

In this  section we consider  a different
$\lambda$-calculus strategy,   strong call-by-name, where
the evaluation of a term may continue
underneath a $\lambda$-abstraction. The main reason is that we wish
to see the impact of this difference on the equivalences  induced by the
 encodings.  Intuitively,  evaluation
underneath a $\lambda$-abstraction is fundamental in the definition
of BTs and therefore we expect that obtaining  the BT equality will be
easier.
However, the LT
equality will  still be predominant: in BTs a
$\lambda$-abstraction is sometimes unobservable, whereas in an
encoding into $\pi$-calculus a
$\lambda$-abstraction always introduces a few prefixes, which
are observable in the most common behavioural equivalences.

The  encoding of strong call-by-name,
from \cite{HMS12}, is  in \afig \ref{f:strongNOTYPE}.
The  encoding behaves similarly to that in
\afig \ref{f:example-lazy}.b;  reduction underneath a `$\lambda$'
is implemented by  exploiting  special wire  processes (such as $\link
qp$). They allow us to split the body $M$ of an abstraction from its
head $\lambda x$; then the  wires make the liaison between the head and the body.
It actually uses the \emph{synchronous}
$\pi$-calculus, because some of the output prefixes have a
continuation. Therefore the encoding also offers us the possibility of
discussing the portability of our conditions to the synchronous
$\pi$-calculus. For this, the only point in which some care is needed
is that in the synchronous $\pi$-calculus, bisimilarity  and expansion
 need
some closure under name substitutions, in the input clause (on the
placeholder name of the input), and the outermost level (that is, before
the bisimulation or expansion game is started) to become  congruence
or precongruence relations. Name substitutions
may be applied following the early, late or open styles. The move from
a style to another one does not affect the results in terms of BTs
and LTs in the paper. We omit the definitions,
 see e.g., \cite{SW01a}.
\begin{figure}[t] 
{\small
\[
\begin{array}{rclcrcl}
\encom{ \lambda  x . M } & \defin &  \abs p \resb{x,q} (\out{p}{x,q}
\para \enco M q)  &\hskip 2cm  &
\encom{x} & \defin &    \abs p (\inp{x}{p'}.  (\link {p'} p))   \\[3pt]
\encom{M N } & \defin & 
\multicolumn{5}{l}{ 
  \abs p \resb{q,r} (\enco{M}{q} \para 
 \inp{q}{x,p'}.(\link{p'}{p} \para  !\out{x}{r}.\enco{N}{r})) \;\;  
 \quad  \mbox{(for  $x$ fresh)}
}\end{array}
\]
$\qquad$ where $\link r q \defin \inp r {y,h} . \out q {y,h}$  
} 
\caption{Encoding of strong call-by-name}\vspace*{-.4cm}
\label{f:strongNOTYPE}
\end{figure}

In short,  for any of the standard  behavioural congruences and
expansion precongruences of the synchronous $\pi$-calculus, the
conditions concerning $\asymp$ and $\leq$ of the theorems in
Section~\ref{ss:sounds_completes}  remain valid.
In Theorem~\ref{t:resSCBN} below,
we continue to
use the symbols $\approx$ and $\expa$ for
bisimilarity and expansion, assuming that these are bisimulation
congruences and   expansion precongruences in any of the common
$\pi$-calculus styles (early, late, open).
Again, in the case of must equivalence
the expansion preorder should be divergence sensitive.
The proof of Theorem~\ref{t:resSCBN} is similar to that of
Theorem~\ref{t:resCBN}. The main difference is that, since in strong
call-by-name the abstraction contexts are not guarded, we have to
adopt the modification in one of the conditions for LTs suggested in
Theorem~\ref{t:sound-strong}. Moreover, for the proof of validity of
$\beta$ rule for $\expa$, we use the following law to reason about
wire processes $\link r q$ 
\Db
(and similarly for   $\expaDiv$); 
see \cite{HMS12,SW01a} for  more discussion:
\De
\begin{itemize}
\item $ \res q ( \link qp | P ) \contr P \sub p q$ provided $p$ does
  not appear free in $P$, and $q$ only appears free in $P$ only once,  in
a subexpression of the form $\out q {\til v}. \nil$.
\end{itemize}
This law is also needed when proving the existence of inverse context (the most involved condition).
The detailed proof of the existence of inverse context is given in  Appendix~\ref{ap:inverse_context_eg}.

\begin{theorem}
\label{t:resSCBN}
The  encoding of Figure~\ref{f:strongNOTYPE}
is fully abstract for LTs when the behavioural equivalence for
the $\pi$-calculus   is $\approx, \may$, or $\mayasy$; and fully abstract
for BTs when the behavioural equivalence is $\must$.
\end{theorem}

Thus
we obtain the BT equality for the must equivalence. Indeed, under
strong call-by-name, all unsolvable terms are divergent.
In contrast with Milner's encoding of Figure~\ref{f:example-lazy}.a,
under asynchronous may equivalence we  obtain  LTs
 because in the encoding of
strong call-by-name the first action of an abstraction is an output,
rather than an input as in Milner's encoding, and outputs are
observable in asynchronous  equivalences.

\section{Types and asynchrony}
\label{s:typ_asy}
 \xx{\large Extend this section and prove Theorem \ref{t:resCBNlin}?}

We show, using Milner's encoding (Figure~\ref{f:example-lazy}.a), that we can
sometimes switch from LTs to BTs by taking into account some simple
\emph{type}  information together with asynchronous forms of behavioural
equivalences.  The type information
needed is the linearity of the parameter name of the encoding (names
$p,q,r$ in Figure~\ref{f:example-lazy}.a). Linearity ensures us that the
external environment can never cause interferences along these names:
if the input capability is used by the process encoding a
$\lambda$-term, then the external environment cannot exercise the same
(competing) capability.  In an asynchronous behavioural equivalence
input prefixes are not directly observable (as discussed earlier for
asynchronous may).

Linear types and asynchrony can easily be incorporated in a
bisimulation congruence by using a contextual form of bisimulation
such as  \emph{barbed congruence} \cite{SW01a}. In this case, barbs  (the observables of
barbed congruence) are only produced by output prefixes (as in
asynchronous may equivalence); and the contexts in which processes may
be tested should respect the type information ascribed to processes
(in particular the linearity mentioned earlier).
We write $\approx_{\rm{bc}}^{\rm{lin,asy}}$ for
 the resulting
asynchronous typed barbed congruence. Using
Theorem~\ref{t:fa}(ii) we obtain:
\begin{theorem} 
\label{t:resCBNlin}
The encoding of Figure~\ref{f:example-lazy}.a
is fully abstract for BTs when the behavioural equivalence for
the $\pi$-calculus   is
$\approx_{\rm{bc}}^{\rm{lin,asy}}$.
\end{theorem}
The auxiliary relation is
still $\expa$; here asynchrony and linearity are not needed.
We give the detailed development of Theorem \ref{t:resCBNlin} in Appendix \ref{ap:property_eg_type}.


%
\section{Conclusions and future work}
\label{sec-conclusion}

In this paper we have studied soundness and completeness
conditions with respect to BTs and LTs for
 encodings of the $\lambda$-calculus into the
$\pi$-calculus. While the conditions have been presented on the
$\pi$-calculus, they can be adapted to some other concurrency
formalisms. For instance, expansion, a key preorder  in our
conditions, can always be extracted from bisimilarity as its
``efficiency'' preorder.  It might be difficult, in contrast, to adapt
our conditions  to sequential languages; a delicate condition, for
instance, appears to be  the one on inversion of
variable contexts.

We have used the conditions to derive tree characterizations for
various encodings and various behavioural equivalences, including
bisimilarity, may and must equivalences, and asynchronous may equivalence.
\begin{xxenv}
Tables (\ref{f:example-lazy}.a), (\ref{f:example-lazy}.b), and (\ref{f:strongNOTYPE})
summarize the results with respect to BTs and LTs
for the encodings and the behavioural equivalences examined in the
paper.
\Db In a table, a check mark means that corresponding result
holds; otherwise \DSb
a
symbol $\crossDS$
indicates that the result is
 false.
\DSe
\De
\DSb
The results in the first column of
 Tables (\ref{f:example-lazy}.a) and (\ref{f:example-lazy}.b)
 appear in the literature \cite{San95lazy,SW01a}.
Concerning the remaining columns and tables,
the results are new, though 
\Db some of them \De
 could have   been obtained with variants of the proofs in  \cite{San95lazy,SW01a}.
The  main contribution of the current paper,  more than
the results themselves, is the  identification   of  some general and
 abstract
conditions that allow one to derive  such results.
\DSe
Some of the check marks are not stated in Theorems \ref{t:resCBN}, \ref{t:resSCBN} and \ref{t:resCBNlin}; they are inferred from the following  facts: soundness for LT implies soundness for BT; completeness for BT implies completeness for LT, since LT equality implies BT equality. 
We recall that  $\;\;\BISI\;\;$ is weak bisimilarity;
 $\may$ is may equivalence;  $\mayasy$  is asynchronous may
 equivalence;   $\must$  is must equivalence; and
 $\approx_{\rm{bc}}^{\rm{lin,asy}}$ is asynchronous barbed congruence
 with the linearity  type constraints on the location names of the
 encoded $\lambda$-terms (i.e., the abstracted name in the encoding of
 a $\lambda $-term).
\end{xxenv}
\Db
The negative results in the tables (i.e., the occurrences of
  symbol   $\crossDS$)
 are consequences of the difference
between LTs and BTs: an encoding that is fully abstract for LTs
cannot be complete for BTs, whereas an
encoding fully abstract for BTs
cannot be sound for LTs.
\De

 \begin{table}[t]
\centering
 \begin{minipage}{0.45\textwidth}%
\renewcommand{\arraystretch}{1.9}
\centering
\begin{tabular}{|c|c|c|c|c|c|c|c|} 
\hline
    &   & $\;\;\BISI\;\;$  & $\may$ & $\mayasy$ &  $\must$ & $\approx_{\rm{bc}}^{\rm{lin,asy}}$\\\hline  
 & \small\,{complete\,}   & \Checkmark  & \Checkmark  & \Checkmark   &
 \Checkmark &
 \Checkmark   \\ \cline{2-7}
\rb{LT} & \small{sound}      & \Checkmark   & \Checkmark  &\crossDS
 &  \Checkmark & \crossDS \\\hline
 & \small{complete} &\crossDS & \crossDS &  \Checkmark    &\crossDS&
 \Checkmark \\\cline{2-7}
\rb{BT} & \small{sound}      & {\Checkmark}   & {\Checkmark}  & \Checkmark   &  {\Checkmark}   &
 \Checkmark\\\hline
\end{tabular}
\vspace*{.1cm}
\caption*{Table \ref{f:example-lazy}.a:  Results for Figure \ref{f:example-lazy}.a} 
\end{minipage}%
\hfill 

\begin{minipage}{0.45\textwidth}%
\renewcommand{\arraystretch}{1.9}
\centering
\begin{tabular}{|c|c|c|c|c|c|} 
\hline
    &   & $\;\;\BISI\;\;$  & $\may$ & $\mayasy$ &  $\must$ \\\hline  
 & \small{\,complete\,} & \Checkmark   & \Checkmark  & \Checkmark  & \Checkmark  \\\cline{2-6}
\rb{LT} & \small{sound}      & \Checkmark    & \Checkmark  & \Checkmark   &  \Checkmark \\\hline
 & \small{complete} & \crossDS & \crossDS & \crossDS  &\crossDS      \\\cline{2-6}
\rb{BT} & \small{sound}      &  {\Checkmark} & {\Checkmark}  & {\Checkmark}  &  {\Checkmark} \\\hline
\end{tabular}
\vspace*{.1cm}
\caption*{Table \ref{f:example-lazy}.b: Results for Figure \ref{f:example-lazy}.b}  
\end{minipage}%
\hfill

 \begin{minipage}{0.45\textwidth}%
\renewcommand{\arraystretch}{1.9}
\centering
\begin{tabular}{|c|c|c|c|c|c|} 
\hline
    &   & $\;\;\approx\;\;$ & $\may$ & $\mayasy$ & $\must$ \\\hline
 & \small{\,complete\,} & \Checkmark  & \Checkmark& \Checkmark   & \Checkmark \\\cline{2-6}
\rb{LT} & \small{sound}      & \Checkmark  & \Checkmark& \Checkmark
                                             & \crossDS \\\hline
 & \small{complete} &\crossDS &\crossDS   &\crossDS&  \Checkmark \\\cline{2-6}
\rb{BT} & \small{sound}      &  {\Checkmark} & {\Checkmark}& {\Checkmark}  & \Checkmark  \\\hline
\end{tabular}
\vspace*{.1cm}
\caption*{Table \ref{f:strongNOTYPE}: Results for Figure \ref{f:strongNOTYPE}} 
\end{minipage}%
\hfill
\end{table}

The proofs of the conditions can often be transported
from a behavioural equivalence to another one, with little or no extra
work (e.g., exploiting containments among equivalences and
preorders). Overall, we found the conditions particularly useful when
dealing with contextual equivalences, such as may and
must equivalences.  It is unclear to us how soundness and completeness
could be proved for them by relying on, e.g., direct characterizations of the
equivalences (such as trace equivalence or forms of acceptance trees) and standard proof techniques for them.

It would be interesting to examine additional  conditions on the behavioural
equivalences of the $\pi$-calculus capable to
 retrieve,
 as
equivalence induced by an encoding,
\begin{xxenv}
that of $\eta$-BTs, or BTs under infinite $\eta$ expansions \cite{Bar84}.
Works on linearity in the $\pi$-calculus, such as \cite{YoshidaHB07}
might be useful; another possibility might be to exploit
\emph{receptive types}, which have a strong impact on the
the sequentiality constraints imposed by input prefixes, see e.g., \cite{Sangiorgi99tapos}.
\end{xxenv}

In the paper we have considered encodings of call-by-name. 
It would be challenging to apply the study to call-by-value; some
preliminary result in this direction has been recently obtained 
\cite{DurierHS18} (based however on different proof techniques, namely
unique solutions of equations \cite{DurierHS17}). 


\sep

\noindent\textsl{Acknowledgements.}
\DSb
We thank the anonymous referees for their constructive comments, which have allowed us to improve the presentations and amend a number of problems in the original document.
\DSe

\bibliographystyle{plain}
\bibliography{process}

\clearpage
\noindent\textbf{\Large Appendix}
\appendix

\section{The proofs for the conditions in Section \ref{ss:sounds_completes}}\label{ap:proof_conditions} 

We first present some auxiliary results.
\subsection{Auxiliary results}

\begin{proposition}\label{p:betatheory}
If  $\qenco$ and $\,\R\,$ validate rule $ \beta$ and $\R$ is a precongruence, then $ M \Rah N$ implies $\encom M \,\R\, \encom N$.
\end{proposition}
\begin{proof}[Proposition \ref{p:betatheory}]
The proof proceeds by \Db  \De induction on the length of $M\wt{}_h N$.
The case when the length is zero is trivial. Now we suppose the length is $n+1$ and show 
\Db that  \De
the result holds.
We know from $M\stm{}_h^{n+1} N$ that
\[
M\stm{}_h^{n} M' \stm{}_h N \, . 
\]
By induction hypothesis, we have
\[\encodingm{M}{}{} \,\R\, \encodingm{M'}{}{}
\]
Next there are several cases to consider with regard to $M' \stm{}_h N$.
\Db \De
\begin{enumerate}[label=(\arabic*)]

\item
\begin{xxenv}
$M'\equiv (M_1M_2)M_3\cdots M_n \stm{}_h M_1'M_3\cdots M_n \equiv N$
in which $M_1M_2$ is a (head) redex and $M_1M_2 \stm{}_h M_1'$.
\DSb
By validity of $\beta$ rule
\DSe
 we know
\[
\encodingm{M_1M_2}{}{} \,\R\, \encodingm{M_1'}{}{} \, .
\]
\DSb
As $\R$ is a precongruence, we can add an arbitrary context, and thus doing
we derive
\DSe
\[
\begin{array}{lcl}
\encodingm{M}{}{} \,\R\, \encodingm{M'}{}{} &\equiv& \qctapp^n[\encodingm{M_1M_2}{}{}, \encodingm{M_3}{}{}, ... , \encodingm{M_n}{}{}] \\
&\,\R\,& \qctapp^n[\encodingm{M_1'}{}{}, \encodingm{M_3}{}{}, ... , \encodingm{M_n}{}{}] \,\equiv\, \encodingm{N}{}{}
\end{array}
\] where $\qctapp^n \DEF \encom { [\cdot]_1 [\cdot]_2 \cdots [\cdot]_{n-1}}$.
\end{xxenv}

\item $M'\equiv \lambda \ve{x}.M_1 \stm{}_h \lambda \ve{x}.M_1' \equiv N$ because $M_1\stm{}_h M_1'$ and $M_1$ is not an abstraction, and $\ve{x}$ denotes $\vect{x}$.
Through similar arguments to case (1), we know
\[
\encodingm{M_1}{}{}\,\R\, \encodingm{M_1'}{}{}
\]
\DSb
and, exploiting the precongruence property of
 $\R$,
\DSe
\[
\begin{array}{lcl}
\encodingm{M}{}{} \,\R\, \encodingm{M'}{}{} &\equiv& C^{\ve{x}}_{\lambda}[\encodingm{M_1}{}{}] \\
&\,\R\,& C^{\ve{x}}_{\lambda}[\encodingm{M_1'}{}{}] \,\equiv\, \encodingm{N}{}{}
\end{array}
\] where $C^{\ve{x}}_{\lambda} \DEF C^{x_1}_{\lambda}[\cdots C^{x_n}_{\lambda}[\cdot] \cdots]$.
\end{enumerate}
This completes the proof.
\end{proof}

\begin{lemma}\label{l:PO}
Suppose that
\begin{enumerate}
\item $\qenco$ and  $\R$ validate rule $ \beta$, and $\R$ is a congruence;
\label{i:POb}
\item whenever $M$ is an unsolvable of order $0$, then $\encom M \,\R\, \encom \Omega$;
\label{i:POd}
\item whenever  $M$ is an unsolvable of order $\infty$, then $\encom M \,\R\, \encom \Omega $.
\label{i:POf}
\end{enumerate}
Then, for any unsolvable $M$ of order $n$ ($n = 0,... \infty$), $\encom {M} \,\R\, \encom{\Omega}$.
\end{lemma}
\begin{proof}[Lemma \ref{l:PO}]
For 
\Db order \De
$\infty$ this is  precisely \ref{i:POf}.
For
\DSb
the remaining cases
\DSe we proceed by induction on $n$. For $n = 0$ this is precisely \ref{i:POd}.
Suppose now $0 < n $ and $M$ is an unsolvable of order $n$. By definition, there is $N$ \st  $M \Rah \lambda x . N$.
\DSb
Thus we have, writing $\Xi$ for  an unsolvable of order $\infty$,
\[
\begin{array}{lclr}
\encom M &\;\R\;& \encom{\lambda x . N} & \qquad\qquad\qquad\qquad\qquad\qquad\mbox{(proposition~\ref{p:betatheory})} \\
 &\equiv& \ctl{\, \encom N \,} & \\
 &\;\R\;& \ctl{\, \encom \Omega \,} & \mbox{(inductive hypothesis)} \\
 &\;\R\;& \ctl{\, \encom \Xi \,} & \mbox{(since $\R$ is a congruence and, by (3), $\Xi \RR
                                   \Omega$)} \\
 &\;\R\;& \encom\Omega & \mbox{($\lambda x.\Xi$ is an unsolvable of order $\infty$)}
\end{array}
\]
\DSe
 which completes the proof.
\end{proof}

\begin{lemma}\label{l:POG}
Suppose that
\begin{enumerate}
\item $\qenco$ and  $\R$ validate rule $ \beta$, and $\R$ is a congruence;
\label{i:POGb}
\item whenever $M,N$ are unsolvable of order $0$, then $\encom M \,\R\, \encom N$;
\label{i:POGd}
\item whenever  $M,N$ are unsolvable of order $\infty$, then $\encom M \,\R\, \encom N$.
\label{i:POGf}
\end{enumerate}
Then, whenever $M,N$ are unsolvable of order $n$ ($n = 0,... \infty$), $\encom {M} \,\R\, \encom{N}$.
\end{lemma}
\begin{proof}[Lemma \ref{l:POG}]
For $\infty$ this is  precisely \ref{i:POGf}.
\XXb For \XXe
\DSb
the remaining cases
\DSe we proceed by induction on $n$. For $n = 0$ this is precisely \ref{i:POGd}.
Suppose now $0 < n $ and $M,N$ are \Db \De  unsolvable of order $n$. By definition, there are $M',N'$ \st  $M \Rah \lambda x . M'$, $N \Rah \lambda x . N'$, and $M',N'$ are unsolvable of order $n{-}1$.
Thus by Proposition~\ref{p:betatheory},
\[
\begin{array}{lclclclr}
\encom{M} & \R & \encom{\lambda x . M'}\equiv \ctl{\, \encom {M'} \,} & \qquad& \ctl{\, \encom {N'} \,} \equiv \encom{\lambda x . N'} & \R & \encom{N} & 
\end{array}
\] Then by induction hypothesis and congruence property of $\R$,
\[
\begin{array}{lclclclr}
\encom{M} & \R & \ctl{\, \encom {M'} \,} & \R& \ctl{\, \encom {N'} \,} & \R & \encom{N} & 
\end{array}
\] Hence $\encom{M} \,\R\, \encom{N}$.
\end{proof}


\subsection{The completeness theorems}


\DSb
\DSe

\begin{proof}[Theorem \ref{t:compNEW}: completeness]
\DSb
We follow the convention that by `condition n',  for \Xab $1 \leqNAT i \leqNAT 6$, \Xae  we mean the corresponding condition in Definition \ref{d:faith}, and by `condition i' or `condition ii' we mean the premise of the corresponding clause in Theorem \ref{t:compNEW}.
\DSe

Assume $\LLT M = \LLT N$, then it follows from the definition of LT equality that one of
the following cases holds 
\Db
 (as usual modulo $\alpha$ conversion).
 \De
\begin{enumerate}[label=(\Roman*)]
\item  $M,N$ are unsolvable of order $0$.
\item $M \Rah   \lambda x . M_1$, $N \Rah   \lambda x . N_1$, and $\LLT {M_1} = \LLT {N_1}$. 
\item $M \Rah    x  M_1\ldots M_n$, $N \Rah    x  N_1 \ldots N_n$, 
and $\LLT {M_i} = \LLT {N_i}$. 
\end{enumerate}
Then we have the following observation.
\begin{itemize}
\item Suppose (I) holds; then, by condition
  \hardcode{6}\conref{i:sLLh}, $\encom M \asymp \encom \Omega$ and
  $\encom N \asymp \encom \Omega$. Thus $\encom{M} \asymp \encom{N}$,
\DSb
because $\asymp$ is an equivalence relation.
\DSe
\item Suppose (II) holds; then, by Proposition~\ref{p:betatheory}, we
  infer $\encom M \geq \encom {\lambda x . M_1} \equiv \ctl{\, \encom{
      M_1} \,}$ and,  in the same way, $\encom N \geq  \ctl{\,
    \encom{N_1} \,}$.
\DSb \DSe
\item Suppose (III) holds; then, by Proposition~\ref{p:betatheory}, we
  infer $\encom M \geq \encom {x M_1\ldots M_n} \equiv
  \ctappn{\encom{M_1}\ldots \encom{M_n}}$ and, in the same way,
  $\encom N \geq \ctappn{\encom{N_1}\ldots \encom{N_n}}$.
\DSb \DSe
\end{itemize}

\DSb
So we are left with cases (II) and (III), which
we handle in the remainder of the proof.
Define $\R$ thus:
\finish{below i have changed so to make sure terms in $\R$ are processes, not abstractions}
\DSe
\[\R \defin \{(\app {\encom M} r , \app {\encom N} r)  \,|\,
\begin{array}[t]{l}
\LLT M = \LLT N ,\\
 \mbox{neither $M$ nor $N$ is unsolvable of order $0$}, \\
\mbox{$r$  fresh}
\}    
\end{array}
 \]
\DSb
In  the remainder we sometimes write $\encom M \,\R\, \encom N$ to mean
$\app {\encom M} r \R \app {\encom N} r$, for some fresh $r$.
We first note that for each $\encom M \,\R\, \encom N$, based on (II) and (III) above and
the following corresponding  observations, one of the following cases is true.
\DSe
\begin{enumerate}[label=(\alph*)]
\item $\encom M \geq \ctl {\encom{M_1}}$, $\encom N \geq \ctl {\encom{N_1}}$ and $(\encom{M_1},\encom{N_1})\in \R$.
\item $\encom M \geq \ctappn {\encom{M_1}, \ldots , \encom{M_n} }$, $\encom N \geq \ctappn {\encom{N_1} , \ldots , \encom{N_n}   }$  and \\
$(\encom{N_i},\encom{N_i})\in \R$  for all $i$.
\end{enumerate}

\DSb
Now, the crux of the proof is to show that
$\R$ is an  \uptoc\ candidate
(Definition~\ref{d:uniCAsyn}), which allows us to conclude $\R\subseteq \asymp$,
 exploiting the property that $\asymp$
validates the \uptoc\ technique, according to condition
\hardcode{3}\conref{i:sLLf}.
This intuitively  will be  possible
because
 $\encom M$ and $\encom
N$ are  related, via the preorder $\geq$, to terms that have a common
context, as shown in (a) and (b) above,
because
 $\leq$ is an expansion relation (condition
 \reff{i:expa} of Definition~\ref{d:faith}),  and because the variable
 and abstraction contexts of $\qenco$ are guarded
 (conditions (i) 
and \hardcode{4}\conref{i:sLLb}), hence the
 first action from terms such as $\ctl {\encom{M_1}}$ and $\ctappn
 {\encom{M_1}, \ldots , \encom{M_n} }$ only consumes the context.
\DSe
\Db 
In both (a) and (b), one does not need to worry about closure under
substitution (of variables) when a hole is underneath an input
prefix, since $\R$ is closed under substitution. That is,
$(\encom{M},\encom{N})\in \R$ implies
$(\encom{M\sigma},\encom{N\sigma})\in \R$, because LT equality is
preserved by variable renaming \cite[Lemma 18.2.6 and Theorem
18.2.7]{SW01a}), 
and because the encoding is uniform (which implies that the free names of 
the encoding of a $\lambda$ term are included in the free variables of that term).
Below are the details for the diagram-chasing requirements.
In the diagrams, the implications of the vertical transitions should
be read from the left to the right.
\DSe
\De
\begin{itemize}
\item If (a) is true, then since $\leq$ is an expansion (condition
  \reff{i:expa}), we have the following diagram
\[
\xymatrix@C=25pt{
\app {\encom M} r \ar@/^1.6pc/@{.}[0,4]|{\R} \ar@{}[r]|-{\displaystyle \geq} \ar[d]^{\mu}
& \app {\ctl {\encom{M_1}}} r \ar[d]^{\widehat{\mu}} &  & \app { \ctl {\encom{N_1}}}
r\ar@{}[r]|-{\displaystyle \leq} \ar[d]^{\widehat{\mu}} & \app { \encom N} r \ar@{=>}[d]^{\widehat{\mu}} \\
S \ar@{}[r]|-{\displaystyle \geq}& C_1[\encom{M_1}] & & C_1[\encom{N_1}]\ar@{}[r]|-{\displaystyle \leq} & T
}
\]
\DSb
The existence of context $C_1$ is due to the fact that $\qctl$ is guarded (condition
(i)), so the action merely consumes the context $\qctl$, and does not affect the term in
the hole. 
In the case $M_1,N_1$ are unsolvable of order $0$, one directly applies $\encom{M_1}
\asymp \encom{N_1}$; otherwise,  $\encom{M_1} \,\R\, \encom{N_1}$ holds.
\DSe

\item If (b) is true
\DSb
then, again because $\leq$ is an expansion, we have
\[
\xymatrix@C=10pt{
 \app {\encom M} r \ar@/^1.9pc/@{.}[0,4]|{\R} \ar@{}[r]|-{\displaystyle \geq} \ar[d]^{\mu}
 &
 \app {
\ctappn {\encom{M_1}, \ldots , \encom{M_n} }} r \ar[d]^{\widehat{\mu}} &  & \app {
 \ctappn {\encom{N_1},\ldots , \encom{N_n}}} r \ar@{}[r]|-{\displaystyle \leq}
\ar[d]^{\widehat{\mu}} &  \app { \encom N} r \ar@{=>}[d]^{\widehat{\mu}} \\
S \ar@{}[r]|-{\displaystyle \geq} & C_2[\encom{M_1}, \ldots , \encom{M_n}] & & C_2[\encom{N_1},\ldots , \encom{N_n}] \ar@{}[r]|-{\displaystyle \leq} & T
}
\]
As in the previous case,  the existence of context $C_2$ is due to the fact that
$\qctappn$ is guarded (condition \hardcode{4}\conref{i:sLLb}).
Moreover, for each $i$, if $M_i,N_i$ are unsolvable of order $0$,
we have $\encom{M_i} \asymp \encom{N_i}$; otherwise, we have $\encom{M_i}
\,\R\, \encom{N_i}$.
\DSe
\end{itemize}


\DSb
This completes the case for LTs.
\DSe
 \sep


\DSb
For  BTs, intuitively, if in a term the head reduction never unveils a variable, then the
term is unsolvable and can be equated to $\Omega$ using the premise of (ii); if head reduction does
unveil a variable, then in the encoding the subterms following the variable are underneath
at least one prefix (because the variable contexts of the encoding are guarded, by
condition \conref{i:sLLb}), and then we are able to apply a reasoning similar to that in
clause (b)  above  for LTs.
Formally,
assume $\BT M = \BT N$. Then, from the definition of BT equality,
 one of the following cases holds: 
\DSe
\begin{enumerate}[label=(\Roman*')]
\item  $M,N$ are unsolvable.
\item $M \Rah  \lambda \ve{x}. xM_1\ldots M_n$, $N \Rah   \lambda \ve{x}. x N_1 \ldots N_n$, 
and $\BT {M_i} = \BT {N_i}$. 
\end{enumerate}
Suppose (I') holds;  then by condition (ii) of this theorem and
\hardcode{(5)}\conref{i:sLLg} and  \conref{i:sLLh} of Definition~\ref{d:faith}, we have
$\encom M \asymp \encom N \asymp \Omega$ by Lemma~\ref{l:PO},
\DSb
which closes the case.
Suppose  now that (II') holds.
We proceed in a similar way to that for LTs.
\DSe
Define $\R'$ as below.
\finish{below, added the $r$}
\[\R' \defin \{( \app {\encom M} r  ,  \app {\encom N}r )  \,|\,
\begin{array}[t]{l}
 \BT M = \BT N , \mbox{neither $M$ nor $N$ is unsolvable}\\
\mbox{$r$ fresh}
\}    
\end{array}
 \]
\DSb
As before, we sometimes write $\encom M \,\R'\, \encom N$ to mean
$\app {\encom M} r \, \R'\,  \app {\encom N} r$, for some fresh $r$.
\DSe
\DSb
\Db
As for LTs, so here we do not need to worry  closure under substitution
of the  `up-to-$\leq$-and-contexts' technique 
because
BT equality, as LT equality, 
 is preserved by  substitution of variables.
\De 
\DSe
For each $\encom M \,\R'\, \encom N$, from (II') and Proposition~\ref{p:betatheory}, we have
\[
\encom{M} \geq \encom{\lambda \ve{x}. xM_1\ldots M_n}, \encom{N} \geq \encom{\lambda \ve{x}. x N_1 \ldots N_n}, \mbox{ and } M_i\,\R'\,N_i
\] Thus
\[
\encom M \geq \ctlxapp{\encom{M_1}\ldots \encom{M_n}}, \encom N \geq \ctlxapp{\encom{N_1}\ldots \encom{N_n}}, \mbox{ and } M_i\,\R'\,N_i
\] for some context $\qctlxapp \DEF \ctlx {x_1} {\,\ctlx {x_2}{\ldots \qctlx {x_l}[\qctappn]\ldots \, }}$ (in which $\ve{x}{=}x_1,...,x_l$).
We know $\qctlxapp$ is guarded thanks to condition \hardcode{4}\conref{i:sLLb},
so some context $C_3$ exists s.t. we have the following chasing diagram
\finish{again, added $r$}
\[
\xymatrix@C=8.5pt{
 \app {\encom M} r \ar@/^2.1pc/@{.}[0,4]|{\R'} \ar@{}[r]|-{\displaystyle \geq}
 \ar[d]^{\mu} &
 \app {
\ctlxapp {\encom{M_1}, \ldots , \encom{M_n} }} r \ar[d]^{\widehat{\mu}} &  &  \app
{\ctlxapp {\encom{N_1},\ldots , \encom{N_n}}} r \ar@{}[r]|-{\displaystyle \leq}
\ar[d]^{\widehat{\mu}} &  \app {\encom N} r \ar@{=>}[d]^{\widehat{\mu}} \\
S \ar@{}[r]|-{\displaystyle \geq} & C_3[\encom{M_1}, \ldots , \encom{M_n}] & & C_3[\encom{N_1},\ldots , \encom{N_n}] \ar@{}[r]|-{\displaystyle \leq} & T
}
\] In the case $M_i,N_i$
\Db  are \De
 unsolvable of any order, one uses $\encom{M_i} \asymp
\encom{N_i}$; for the
\DSb
remaining
 cases, one applies $\encom{M_i} \,\R'\, \encom{N_i}$.

We have thus shown  that $\R'$ is an  \uptoc\ candidate,  and we can
finally conclude  by condition \hardcode{3}\conref{i:sLLf} that $\R'\subseteq \asymp$.
\end{proof}
\sep
\DSe

\DSb
We conclude by presenting the proof of Theorem~\ref{t:sound-strong}, which gives
\Db us \De
 some
alternative condition for completeness (which  also yields an alternative condition  for
full abstraction). Precisely, Theorem~\ref{t:sound-strong}
replaces the condition that the abstraction contexts  be guarded with the requirement that
\[
\mbox{$ $} \hfill \hskip 2cm
\mbox{$M,N $ unsolvable of order $\infty$
implies $\encom M \asymp  \encom N$.  }
\hfill \hskip 2cm  \mbox{$(*)$}
\]
\DSe

\begin{proof}[Theorem 4: alternative completeness conditions]
If $\LLT M = \LLT N$, then one of the following holds.
\begin{enumerate}[label=(\Roman*)]
\item $M,N$ are unsolvable of order $m$ ($m=0,...,\infty$).
\item $M \Rah \lambda \til x . x M_1\ldots M_n$, and $N \Rah \lambda \til x . x
  N_1\ldots N_n$, and $\LLT {M_i} = \LLT {N_i}$ (in which
$\ve{x} = x_1,...,x_l$, 
\Db for some $l$, and  \De
 $i=1,...,n$). 
\end{enumerate}
\DSb
In case (I), by Lemma \ref{l:POG} (using condition
\hardcode{6}\conref{i:sLLh}, and condition $(*)$ in the statement of
Theorem~\ref{t:sound-strong}),
we derive
 $\encom{M} \asymp \encom{N}$, as desired.
\DSe
In case (II), let
\[\R \defin \{(\app{ \encom M} r , \app{\encom N} r)  \,|\,
\begin{array}[t]{l}
 \LLT M = \LLT N , \mbox{neither $M$ nor $N$ is unsolvable } \\
\mbox{$r$ fresh}
\}    
\end{array}
 \]
We sometimes write $\encom M \,\R\, \encom N$ to mean
$\app {\encom M} r \, \R\,  \app {\encom N} r$, for some fresh $r$.
\DSb
We prove that the relation is an  \uptoc\ candidate
(Definition~\ref{d:uniCAsyn}), which allows us to conclude $\R\subseteq \asymp$,
 by condition \hardcode{3}\conref{i:sLLf}.
\Db As in the proof of Theorem~\ref{t:compNEW}, so here 
 relation $\R$ is closed under name substitutions, which  is needed
for application of condition \hardcode{3}\conref{i:sLLf}.
\De 
\DSe
For each $\encom M \,\R\, \encom N$, from (ii) and Proposition~\ref{p:betatheory}, we have
\[
\encom{M} \geq \encom{\lambda \ve{x}. xM_1\ldots M_n}, \encom{N} \geq \encom{\lambda \ve{x}. x N_1 \ldots N_n}, \mbox{ and } M_i\,\R\,N_i
\] Thus
\[
\encom M \geq \ctlxapp{\encom{M_1}\ldots \encom{M_n}}, \encom N \geq \ctlxapp{\encom{N_1}\ldots \encom{N_n}}, \mbox{ and } M_i\,\R\,N_i
\] for some context $\qctlxapp \DEF \ctlx {x_1} {\,\ctlx {x_2}{\ldots \qctlx {x_l}[\qctappn]\ldots \, }}$.
A key point here is that $\qctlxapp$ is guarded thanks to condition \hardcode{4}\conref{i:sLLb}.
Thus some context $C_4$ exists s.t. the following chasing diagram holds
\finish{added the $r$}
\[
\xymatrix@C=10pt{
\app{\encom M} r \ar@/^2.1pc/@{.}[0,4]|{\R} \ar@{}[r]|-{\displaystyle \geq} \ar[d]^{\mu} &
\app{
\ctlxapp {\encom{M_1}, \ldots , \encom{M_n} }} r \ar[d]^{\widehat{\mu}} &  &
\app{
 \ctlxapp {\encom{N_1},\ldots , \encom{N_n}}} r \ar@{}[r]|-{\displaystyle \leq}
\ar[d]^{\widehat{\mu}} &
\app{
\encom N} r \ar@{=>}[d]^{\widehat{\mu}} \\
S \ar@{}[r]|-{\displaystyle \geq} & C_4[\encom{M_1}, \ldots , \encom{M_n}] & & C_4[\encom{N_1},\ldots , \encom{N_n}] \ar@{}[r]|-{\displaystyle \leq} & T
}
\] In the case $M_i,N_i$
\DSb
are unsolvable \Db then they are unsolvable of the same order and 
\De we have $\encom{M_i} \asymp \encom{N_i}$; for the
remaining cases, we have
\DSe
 $\encom{M_i} \,\R\, \encom{N_i}$.
\DSb
This completes the proof.
\DSe
\end{proof}


\subsection{The soundness theorem}
\begin{proof}[Theorem \ref{t:soundNEW}: soundness]

We define
\[
\R\DEF \{(M,N) \,|\; \encodingm{M}{}{} \asymp \encodingm{N}{}{} \}
\] and show that $\mathcal{R}$ implies LT equality. To that aim, it suffices to prove that, for any $M \,\R\, N$ (i.e., $M\asymp N$), the following properties hold.
\begin{enumerate}
\item If $M$ is unsolvable of order $0$, then so is $N$;
\item If $M\wt{}_h \lambda x.M_1$, then $N\wt{}_h \lambda x.N_1$ and $M_1 \,\mathcal{R}\, N_1$;
\item If $M\wt{}_h xM_1\cdots M_n$, then $N\wt{}_h xN_1\cdots N_n$ and $M_i \,\mathcal{R}\, N_i$ for every $i=1,...,n$.
\end{enumerate}
The proof proceeds by the case analysis below.
\DSb
Similarly to what \Xab is \Xae done before, here
 by `condition n',  for \Xab $1 \leqNAT i \leqNAT 6$, \Xae we mean the corresponding condition in Definition \ref{d:resp}, and by `condition i' or `condition ii' we mean the premise of the corresponding clause in Theorem \ref{t:soundNEW}.
\DSe
\begin{enumerate}
\item $M$ is unsolvable of order $0$.  By condition \hardcode{5}\conref{con:lls:e}, we know $\encodingm{M}{}{} \asymp \encodingm{\Omega}{}{} $, then since $M\asymp N$
\[
\encodingm{\Omega}{}{} \asymp \encodingm{N}{}{} 
\]
By condition \hardcode{6}\conref{con:lls:f} and condition (i), it must
be that $N$ is unsolvable of order $0$. This is because if not, two
cases are possible: (1) $N$ has order other than $0$; (2) $N$ (head)
reduces to $\lambda \ve{y}.z\ve{N}$. Either case would contradict the
conditions
\DSb
(conditions (6) and (i),  using Proposition \ref{p:betatheory}).
\DSe

\item $M\wt{}_h \lambda x.M_1$.  By Proposition \ref{p:betatheory} (using conditions \hardcode{1}\conref{con:lls:a} and \hardcode{4}\conref{con:lls:d}),
\[
\encodingm{M}{}{} \geq \encodingm{\lambda x.M_1}{}{}
\]
Then we know from condition \hardcode{2}\conref{con:lls:b} that $\encodingm{\lambda x.M_1}{}{} \asymp \encodingm{M}{}{}$.
So 
\[
\encodingm{N}{}{} \asymp \encodingm{\lambda x.M_1}{}{}
\]
Now from condition \hardcode{6}\conref{con:lls:f} and condition (i) of this theorem, 
it must be that $N$
\DSb
head reduces to
\DSe
 $\lambda x. N_1$, for some $N_1$, so 
\begin{equation}\label{on-abstraction-context-equ2}
\ctl{\encodingm{M_1}{}{}} \equiv \encodingm{\lambda x.M_1}{}{} \asymp \encodingm{\lambda x. N_1}{}{} \equiv \ctl{\encodingm{N_1}{}{}}
\end{equation}
By condition \hardcode{7}\conref{con:lls:gh}, we suppose $D$ is the existing context as stated in Definition \hardcode{8}\defref{d:inverse_conext}.
Then we have 
\[
\xymatrix{
D[C^x_\lambda[\encodingm{M_1}{}{}]] \ar@{}[r]|-{\displaystyle \asymp}& D[C^x_\lambda[\encodingm{N_1}{}{}] \\
(\res{\ve{b}})(\out{a}{\ve{c}}\para b(z).\encodingm{M_1}{}{}\lrangle{z}) \ar@{}[u]|-{\displaystyle \rotatebox[origin=c]{90}{$\leq$}} & (\res{\ve{b}})(\out{a}{\ve{c}}\para b(z).\encodingm{N_1}{}{}\lrangle{z}) \ar@{}[u]|-{\displaystyle \rotatebox[origin=c]{90}{$\leq$}}
}
\] where $a,b,z$ fresh, and $b\in \ve{b}\subseteq{\ve{c}}$; we recall that the encoding of a $\lambda$-term is an abstraction of the $\pi$-calculus. 
Thus
\[
(\res{\ve{b}})(\out{a}{\ve{c}}\para b(z).\encodingm{M_1}{}{}\lrangle{z}) \asymp (\res{\ve{b}})(\out{a}{\ve{c}}\para b(z).\encodingm{N_1}{}{}\lrangle{z})
\] By condition \hardcode{3}\conref{con:lls:c},
\[
\encodingm{M_1}{}{}\lrangle{z} \asymp \encodingm{N_1}{}{}\lrangle{z}
\] Thus
\DSb
by Lemma~\ref{l:plain}(2)
\DSe
\[
\encodingm{M_1}{}{} \asymp \encodingm{N_1}{}{}
\]
Hence in summary, $N\wt{}_h \lambda x.N_1$ and $M_1 \,\mathcal{R}\, N_1$.

\item $M\wt{}_h xM_1\cdots M_n$. By Proposition \ref{p:betatheory} (using conditions \hardcode{1}\conref{con:lls:a} and \hardcode{4}\conref{con:lls:d}),
\[
\encodingm{M}{}{} \geq \encodingm{xM_1\cdots M_n}{}{}
\]
Then we know from condition \hardcode{2}\conref{con:lls:b} that $\encodingm{M}{}{} \asymp \encodingm{xM_1\cdots M_n}{}{}$.
So 
\[
\encodingm{N}{}{} \asymp \encodingm{xM_1\cdots M_n}{}{}
\]
Now from condition \hardcode{6}\conref{con:lls:f} and condition (i) of this theorem, 
it must be that $N$ derives (i.e., head reduces to) $xN_1\cdots N_n$ for some $N_1,...,N_n$, so by Proposition \ref{p:betatheory} we have 
\begin{equation}\label{on-variable-context-equ2}
\ctappn{\encodingm{M_1}{}{},...,\encodingm{M_n}{}{}} {\equiv} \encodingm{xM_1\cdots M_n}{}{} \asymp \encodingm{xN_1\cdots N_n}{}{} {\equiv} \ctappn{\encodingm{N_1}{}{},...,\encodingm{N_n}{}{}}
\end{equation}
Then by condition \hardcode{7}\conref{con:lls:gh}, we suppose $D_i$ ($i=1,...,n$) is the existing context as stated in Definition \hardcode{8}\defref{d:inverse_conext}. So we have
\[
\xymatrix{
D_i[\ctappn{\encodingm{M_1}{}{},...,\encodingm{M_n}{}{}}] \ar@{}[r]|-{\displaystyle \asymp}& D_i[\ctappn{\encodingm{N_1}{}{},...,\encodingm{N_n}{}{}}] \\
(\res{\ve{b}})(\out{a}{\ve{c}}\para b(z).\encodingm{M_i}{}{}\lrangle{z}) \ar@{}[u]|-{\displaystyle \rotatebox[origin=c]{90}{$\leq$}} & (\res{\ve{b}})(\out{a}{\ve{c}}\para b(z).\encodingm{N_i}{}{}\lrangle{z}) \ar@{}[u]|-{\displaystyle \rotatebox[origin=c]{90}{$\leq$}}
}
\] where $a,b,z$ fresh, and $b\in \ve{b}\subseteq{\ve{c}}$.
Thus
\[
(\res{\ve{b}})(\out{a}{\ve{c}}\para b(z).\encodingm{M_i}{}{}\lrangle{z}) \asymp (\res{\ve{b}})(\out{a}{\ve{c}}\para b(z).\encodingm{N_i}{}{}\lrangle{z})
\] By condition \hardcode{3}\conref{con:lls:c},
\[
\encodingm{M_i}{}{}\lrangle{z} \asymp \encodingm{N_i}{}{}\lrangle{z}
\] Thus \DSb
by Lemma~\ref{l:plain}(2)
\DSe
\[
\encodingm{M_i}{}{} \asymp \encodingm{N_i}{}{}
\] Hence in summary, $N\wt{}_h \, xN_1\cdots N_n$ and $M_i \,\mathcal{R}\, N_i$ for every $i=1,...,n$.
\end{enumerate}
\DSb
This completes  the proof for LTs.
\DSe
\sep

For the BT case, we define
\[
\mathcal{S}\DEF \{(M,N) \,|\; \encodingm{M}{}{} \asymp \encodingm{N}{}{} \}
\] and show that $\mathcal{R}$ implies BT equality. To this end, we prove that, for any $M \,\R\, N$ (i.e., $M\asymp N$), the following properties hold.
\begin{enumerate}
\item If $M$ is unsolvable of order $n$ ($n=0,...,\infty$), then $N$ is unsolvable of order $m$ ($m=0,...,\infty$).
\item If $M\wt{}_h\, \lambda \ve{x}.xM_1\cdots M_n$, then $N\wt{}_h\, \lambda \ve{x}.xN_1\cdots N_n$ and $M_i \,\mathcal{R}\, N_i$ for every $i=1,...,n$.
\end{enumerate}
\DSb
The proof proceeds by the following  case analysis.
\DSe
\begin{enumerate}
\item 
$M$ is unsolvable of order $n$ ($n=0,...,\infty$). By Lemma \ref{l:PO} (using condition \hardcode{4}\conref{con:lls:d}, condition \hardcode{5}\conref{con:lls:e}, and condition (ii).(a) 
of this theorem), $\encom{M} \asymp \encom{\Omega}$. Since $\encodingm{M}{}{} \asymp \encodingm{N}{}{}$, we have
\[
\encom{N} \asymp \encom{\Omega}
\] Thus $N$ must be unsolvable of some order, because if not,
\DSb
a
contradiction would arise by appealing to condition
\hardcode{6}\conref{con:lls:f} and  condition (ii).(b) 
of this theorem, and to  Proposition \ref{p:betatheory}.
\DSe

\item 
$M\wt{}_h\, \lambda \ve{x}.yM_1\cdots M_n$. This case can be dealt
\DSb
with in a similar way to that for LTs, by combining cases 2 and 3
there; here one uses  condition
\hardcode{7}\conref{con:lls:gh} several times  (precisely, the length of
$\ve{x}$ plus one): one for a variable context and the others for abstraction
contexts. Also the condition (ii).(b) 
is used 
when determining the shape of $N$. \qedhere
\DSe
\end{enumerate}
\end{proof}


\section{The `inverse context' property of the encodings}\label{ap:inverse_context_eg} 

\begin{xxenv}
Lemmas \ref{l:contexts_inverse:1}, \ref{l:contexts_inverse:2}, \ref{l:contexts_inverse:3} provide the inverse context properties of the examples in Section \ref{sec-examples-cbn}, \ref{sec-examples-scbn}. 
In each of them, we first give the form of the inverse context, then exemplify how it works when fed with encodings of $\lambda$-terms, which is the frequent case (see Theorem \ref{t:soundNEW}), though generic abstraction can be used in the same way. 
\end{xxenv}
\DSb
We recall that $\sim$ is strong bisimilarity (Section~\ref{sec-background}).
\DSe 

\begin{lemma}[On the first call-by-name encoding, \afig \ref{f:example-lazy}.a]\label{l:contexts_inverse:1}
\

\begin{enumerate}
\item The abstraction contexts of $\encodingm{\,}{}{}$ have inverse with respect to $\contrdiv$;
\item The variable contexts of $\encodingm{\,}{}{}$ have inverse with respect to $\contrdiv$. 
\end{enumerate}
\end{lemma}
\begin{proof} \xx{\rc{please check if possible}}

\noindent\textbf{1}. 
The abstraction contexts \Db \De  are defined by
\[C^x_\lambda \equiv \abs p \inp{ p}{x, q}.  [\cdot]\lrangle{q}
\] 
We define context $D$ as below. 
\[
D \DEF (\res{r,b})(\out{a}{b}\para b(r_1).([\cdot]\lrangle{r} \para \out{r} {x, r_1})) 
\]
Then 
\[D[C^x_\lambda[\encodingm{M}{}{}]] \;\contrdiv\; (\res{r,b})(\out{a}{b}\para b(r_1).(\encodingm{M}{}{}\lrangle{r_1} )) \sim (\res{b})(\out{a}{b}\para b(r_1).(\encodingm{M}{}{}\lrangle{r_1} ))
\]

\noindent\textbf{2}. 
We know from the encoding that  the variable contexts are defined as below.
\[
\begin{array}{lcl}
C^{x,n}_{var} &\equiv& \encodingm{x[\cdot]_1\cdots [\cdot]_n}{}{} \\
&\equiv & \abs p \res{r, y} \Big({\enco{x[\cdot]_1\cdots [\cdot]_{n-1}}{r}  |  \out r {y, p} |  !  \inp y q .  \enco {[\cdot]_n} q}\Big) \quad y \mbox{ fresh}
\end{array}
\] If $n=0$, there is nothing to prove. Suppose $n \geqslant 1$. 
By \cite{San95lazy}(Lemma 5.2), we know
\DSb
\XXb
\[
\begin{array}{lcl}
C^{x,n}_{var}[ \encodingm{M_1}{},\cdots, \encodingm{M_n}{}]\lrangle{r_n} &\equiv& \encodingm{xM_1\cdots M_n}{}{}\lrangle{r_n} \\
&\sim& \res{r_0} (\out x {r_0} \para \mathcal{O}\lrangle{r_0,r_n,\encodingm{M_1}{}{},...,\encodingm{M_n}{}{}}) 
\end{array}
\]
\Db  where $r_1,...,r_{n-1},x_1,...,x_n,q \mbox{ are fresh}$ and  \De
\[
\begin{array}{lcl}
\mathcal{O}\lrangle{r_0,r_n,F_1,...,F_n} &\DEF& 
 \res{r_1,...,r_{n-1},x_1,...,x_n}
\\ && \Big( \begin{array}[t]{l}
 \out {r_0} {x_1,r_1} \para \cdots \para \out {r_{n-1}} {x_n,r_n} \\
 \para !\inp {x_1} q.F_1\lrangle{q} \para \cdots\para  !\inp {x_n} q.F_n\lrangle{q}  \Big)
\end{array}
\end{array}
\]
\DSe 
\XXe
\Db \De
We need to find the contexts $\{D^{x,n}_i \,|\, 1\leqslant i \leqslant n\}$ in which
\[ D^{x,n}_i[C^{x,n}_{var}[\encodingm{M_1}{}{},...,\encodingm{M_n}{}{}]] \;\contrdiv\; (\res{\ve{b}})(\out{a}{\ve{c}}\para b(z).\encodingm{M_i}{}{}\lrangle{z}) \quad (b\in \ve{b}\subseteq{\ve{c}})
\] 
\XXb
The context $D^{x,n}_i$ takes the shape below ($0<j<i-1; a,b,z$ fresh).
\[
\begin{array}{lcl}
D^{x,n}_i &\DEF& (\res{r_n,x,b})([\cdot]\lrangle{r_n} \para \\
&& \inp x {r_0}. \inp {r_0} {x_1,r_1}.....\inp {r_j} {x_{j+1},r_{j+1}}.....\inp {r_{i-1}} {x_i,r_i'} . (\out{a}{x,b} \para \\
&& \inp{b}{z}.\out {x_i} {z}))
\end{array}
\] It can be observed that
\[
\begin{array}{rl}
 & D^{x,n}_i[C^{x,n}_{var}[\encodingm{M_1}{}{},...,\encodingm{M_n}{}{}]] \\
\sim & (\res{r_n,x,b})(\res{r_0}(\out x {r_0} \para \mathcal{O}\lrangle{r_0,r_n,\encodingm{M_1}{}{},...,\encodingm{M_n}{}{}}) \para \\
 & \inp x {r_0}. \inp {r_0} {x_1,r_1}.....\inp {r_j} {x_{j+1},r_{j+1}}.....\inp {r_{i-1}} {x_i,r_i'} . (\out{a}{x,b} |\inp{b}{z}.\out {x_i} {z})) \\
\contrdiv 
& (\res{r_n,r_i,x,x_i,b})(\mathcal{O}\lrangle{r_i,r_n,\encodingm{M_{i+1}}{}{},...,\encodingm{M_n}{}{}} \para \\
 & !\inp {x_i} q.\encodingm{M_i}{}{}\lrangle{q}\para \out{a}{x,b} \para \inp{b}{z}.\out {x_i} {z}) \\
\sim & (\res{x,x_i,b})(!\inp {x_i} q.\encodingm{M_i}{}{}\lrangle{q}\para \out{a}{x,b} \para \inp{b}{z}.\out {x_i} {z}) \\
\sim & (\res{x,x_i,b})(\out{a}{x,b} \para \inp{b}{z}.(!\inp {x_i} q.\encodingm{M_i}{}{}\lrangle{q}\para \out {x_i} {z})) \\
\contrdiv & (\res{x,b})(\out{a}{x,b} \para \inp{b}{z}.\encodingm{M_i}{}{}\lrangle{z}) 
\end{array}
\] where the first occurrence of 
$\contrdiv$ subsumes  ($i+1$)  internal 
$\tau$ actions. So we are done.
\XXe
\end{proof} 


\begin{lemma}[On the second call-by-name encoding, \afig \ref{f:example-lazy}.b]\label{l:contexts_inverse:2}
\

\begin{enumerate}
\item The abstraction contexts of $\encodingm{\,}{}{}$ have inverse with respect to $\contrdiv$;
\item The variable contexts of $\encodingm{\,}{}{}$ have inverse with respect to $\contrdiv$. 
\end{enumerate}
\end{lemma}
\begin{proof} \xx{\rc{please check if possible}} 
\
\Db 
\finish{in this appendix, i have replaced the strong and weak transitions, because the
  inverse property talks about behavioural relations (something like 5 occurrences)} 
\De

\noindent\textbf{1}. 
The abstraction contexts are
\[C^x_\lambda \equiv  \abs p \res v (\out{ p}{v} \para \inp {v} {x,q} . [\cdot]\lrangle{q})
\] 
We define context $D$ thus: 
\[
D \DEF (\res{r,b}) (\out{a}{b}\para \inp{b}{r_1}.([\cdot]\lrangle{r} \para \inp {r} v.\out {v} {x, r_1}) ) 
\]
It then holds \Db that \De
\[
D[C^x_\lambda[\encodingm{M}{}{}]] \;\contrdiv\; (\res{r,b}) (\out{a}{b}\para \inp{b}{r_1}.(\encodingm{M}{}{}\lrangle{r_1}) ) \sim (\res{b}) (\out{a}{b}\para \inp{b}{r_1}.(\encodingm{M}{}{}\lrangle{r_1}) )
\]

\noindent\textbf{2}. In the encoding the variable contexts are defined as: 
\[
\begin{array}{lcl}
 C^{x,n}_{var} &\equiv& \encodingm{x[\cdot]_1\cdots [\cdot]_n}{}{} \\ 
& \equiv & \abs p \res{r} \Big(\enco{x[\cdot]_1\cdots [\cdot]_{n-1}}{r}  |  \inp r v. \res{y} (\out v {y, p} |  !  \inp y q .  \enco  {[\cdot]_n} q)\Big),  \;\;\mbox{ $y$ fresh.}
\end{array}
\] Suppose $n \geqslant 1$, since there is nothing to prove if $n=0$.  
By an inductive analysis similar to that in Lemma \ref{l:contexts_inverse:1}, we have
\XXb
\[
\begin{array}{lcl}
C^{x,n}_{var}[ \encodingm{M_1}{},\cdots, \encodingm{M_n}{}]\lrangle{r_n} &\equiv&   \encodingm{xM_1\cdots M_n}{}{}\lrangle{r_n} \\
&\sim& \res{r_0} (\out x {r_0} \para \mathcal{O}\lrangle{r_0,r_n,\encodingm{M_1}{}{},...,\encodingm{M_n}{}{}})
\end{array}
\]
where 
\Db 
$r_1,...,r_{n-1},v_0,...,v_{n-1},x_1,...,x_n,q \mbox{ are fresh}$ and 
\[
\begin{array}{lcl}
\mathcal{O}\lrangle{r_0,r_n,F_1,...,F_n} &\DEF& \res{r_1,...,r_{n-1},x_1,...,x_n}
\\ &&\Big( \begin{array}[t]{l}
\inp {r_0} {v_0}.\out {v_0} {x_1,r_1} \para \cdots \para \\
 \inp {r_{n-1}} {v_{n-1}}.\out {v_{n-1}} {x_n,r_n} \para \\
  !\inp {x_1} q.F_1\lrangle{q} \para \cdots\para  !\inp {x_n} q.F_n\lrangle{q}  \Big)
\end{array}
\end{array}
\]
\De 
\XXe
We need to design the contexts $\{D^{x,n}_i \,|\, 1\leqslant i \leqslant n\}$ in which
\[ D^{x,n}_i[C^{x,n}_{var}[\encodingm{M_1}{}{},...,\encodingm{M_n}{}{}]] \;\contrdiv\; (\res{\ve{b}})(\out{a}{\ve{c}}\para b(z).\encodingm{M_i}{}{}\lrangle{z}) \quad (b\in \ve{b}\subseteq{\ve{c}})
\] The context $D^{x,n}_i$ is defined 
\Db thus, for 
$0<j<i-1$, and $ a,b,z$ fresh.
\XXb
\[
\begin{array}{lcl}
D^{x,n}_i &\DEF& (\res{r_n,x,v_0,...,v_{i-1},b})
\\ && ( \begin{array}[t]{l}
[\cdot]\lrangle{r_n} \para \inp x {r_0}. \out {r_0} {v_0} \para \inp {v_0} {x_1,r_1}.\out {r_1}{v_1} \para \\
  ...\para \inp {v_{j-1}}{x_j,r_j}.\out {r_j} {v_j} \para \\
  ...\para \inp {v_{i-2}}{x_{i-1},r_{i-1}} .\out {r_{i-1}}{v_{i-1}}\para \inp {v_{i-1}} {x_i,r_i'} . (\out{a}{x,b} \para \inp{b}{z}.\out {x_i} {z})) 
\end{array}
\end{array}
\] It holds that
\begingroup
\addtolength{\jot}{-0.2em}
\begin{align*}
 & D^{x,n}_i[C^{x,n}_{var}[\encodingm{M_1}{}{},...,\encodingm{M_n}{}{}]] \\
\sim~ & (\res{r_n,x,v_0,...,v_{i-1},b}, {r_0})
\\ &
( \begin{array}[t]{l}
\out x {r_0} \para
    \mathcal{O}\lrangle{r_0,r_n,\encodingm{M_1}{}{},...,\encodingm{M_n}{}{}}) \para \\ 
  \inp x {r_0}. \out {r_0} {v_0} \para \inp {v_0} {x_1,r_1}.\out {r_1}{v_1} \para \\
  ...\para \inp {v_{j-1}}{x_j,r_j}.\out {r_j} {v_j} \para \\
  ...\para \inp {v_{i-2}}{x_{i-1},r_{i-1}} .\out {r_{i-1}}{v_{i-1}}\para \inp {v_{i-1}} {x_i,r_i'} . (\out{a}{x,b} \para \inp{b}{z}.\out {x_i} {z}) \\
\end{array} \\
\contrdiv~ 
& (\res{r_n,x,x_i,r_i,b})(\mathcal{O}\lrangle{r_i,r_n,\encodingm{M_{i+1}}{}{},...,\encodingm{M_n}{}{}} \para \\
& !\inp {x_i} q.\encodingm{M_i}{}{}\lrangle{q} \para \out{a}{x,b} \para \inp{b}{z}.\out {x_i} {z}) \\  
\sim~ & (\res{x,x_i,b})(!\inp {x_i} q.\encodingm{M_i}{}{}\lrangle{q} \para \out{a}{x,b} \para \inp{b}{z}.\out {x_i} {z}) \\
\sim~ & (\res{x,x_i,b})(\out{a}{x,b} \para \inp{b}{z}.(!\inp {x_i} q.\encodingm{M_i}{}{}\lrangle{q} \para \out {x_i} {z})) \\
\contrdiv~ & (\res{x,b})(\out{a}{x,b} \para \inp{b}{z}.\encodingm{M_i}{}{}\lrangle{z}). \tag*{\qEd}
\end{align*}
\endgroup
\De
\XXe
\def\popQED{}
\end{proof} 


\begin{lemma}[On the strong call-by-name encoding, \afig \ref{f:strongNOTYPE}]\label{l:contexts_inverse:3} 
\

\begin{enumerate}
\item The abstraction contexts of $\encodingm{\,}{}{}$ have inverse with respect to $\contrdiv$;
\item The variable contexts of $\encodingm{\,}{}{}$ have inverse with respect to $\contrdiv$. 
\end{enumerate}
\end{lemma}
\begin{proof} \xx{\rc{please check if possible}}

As noted, the following property (which admits a routine reasoning) is used in the proof of this lemma.
\begin{equation}\label{scbn-prop-wire}
\res r (\link{r}{p} | \enco M r) \contrdiv \enco M p
\end{equation} 
Below we cope with each clause of the lemma.

\noindent\textbf{1}.
The abstraction context is
\[
C^x_\lambda \equiv \abs p \res{x,q} (\out{p}{x,q} \para [\cdot]\lrangle{q})
\]
We define 
\[
D \DEF (\res{r,b}) ([\cdot]\lrangle{r} \para \inp {r}{x,q}.(\out{a}{x,b}\para \inp{b}{r_1}.(\link{q}{r_1})) )
\]
\XXb
We then have
\[
\begin{array}{rl}
 & D[C^x_\lambda[\encodingm{M}{}{}]] \\
 \equiv & (\res{r,b}) ((\res{x,q}) (\out{r}{x,q} \para \encodingm{M}{}{}\lrangle{q}) \para \inp {r}{x,q}.(\out{a}{x,b}\para \inp{b}{r_1}.(\link{q}{r_1})) ) \\
\contrdiv
 & (\res{r,b}) ((\res{x,q}) (\encodingm{M}{}{}\lrangle{q} \para \out{a}{x,b}\para \inp{b}{r_1}.(\link{q}{r_1}) )) \\
\sim & (\res{b,x,q}) (\out{a}{x,b} \para \encodingm{M}{}{}\lrangle{q} \para  \inp{b}{r_1}.(\link{q}{r_1}) ) \\
\sim & (\res{b,x,q}) (\out{a}{x,b} \para \inp{b}{r_1}.(\encodingm{M}{}{}\lrangle{q} \para \link{q}{r_1}) ) \\
\contrdiv & (\res{b,x}) (\out{a}{x,b}\para \inp{b}{r_1}.\encodingm{M}{}{}\lrangle{r_1})
\end{array}
\]  where property (\ref{scbn-prop-wire}) is used. 
\XXe

\noindent\textbf{2}. The variable context in the encoding is defined as 
\[
\begin{array}{lcl}
 C^{x,n}_{var} &\equiv& \encodingm{x[\cdot]_1\cdots [\cdot]_n}{}{} \\ 
& \equiv & \abs p \res{q,r} (\encodingm{x[\cdot]_1\cdots [\cdot]_{n-1}}{}{}\lrangle{q} \para  \inp{q}{y,p'}.(\link{p'}{p} \para !\out{y}{r}.\encodingm{[\cdot]_n}{}{}\lrangle{r})) \; \mbox{ ($y$ fresh)} 
\end{array}
\] Suppose $n\geqslant 1$ (nothing to prove if $n=0$). By an inductive analysis similar to that in Lemma \ref{l:contexts_inverse:1}, we have
\DSb
\XXb
\[
\begin{array}{l}
C^{x,n}_{var}[ \encodingm{M_1}{},\cdots,
  \encodingm{M_n}{}]\lrangle{r_n}  \equiv \\
   \encodingm{xM_1\cdots M_n}{}{}\lrangle{r_n} 
\sim \res{q_0} (\inp{x}{p_0'}.  (\link {p_0'}{q_0}) \para \mathcal{O}\lrangle{q_0,q_n,\encodingm{M_1}{}{},...,\encodingm{M_n}{}{}}) \\
\end{array}
\]
\XXe
where 
\Db 
$q_1,...,q_{n-1},r_1,...,r_n,x_1,...,x_n$ are fresh and 
\De \[ 
\begin{array}{lc}
 \;\;\mathcal{O}\lrangle{q_0,q_n,F_1,...,F_n} \DEF \\
 \res{q_1,...,q_{n-1},r_1,...,r_{n},x_1,...,x_n}
\\  
\Big( \begin{array}[t]{l}
 \inp {q_0} {x_1,p_1'}.(\link{p_1'}{q_1}\para !\out {x_1} {r_1}.\encodingm{M_1}{}{}\lrangle{r_1}) \para \\
  \cdots \para \inp {q_{i-1}} {x_i,p_{i}'}.(\link{p_{i}'}{q_{i}}\para !\out {x_i} {r_i}.\encodingm{M_i}{}{}\lrangle{r_i}) \\
 \cdots \para \inp {q_{n-1}} {x_n,p_{n}'}.(\link{p_{n}'}{q_{n}}\para !\out {x_n} {r_n}.\encodingm{M_n}{}{}\lrangle{r_n})  \Big) \qquad (i=1,...,n)
\end{array}
\end{array}
\] 
\DSe 
We need to design the contexts $\{D^{x,n}_i \,|\, 1\leqslant i \leqslant n\}$ in which
\[ D^{x,n}_i[C^{x,n}_{var}[\encodingm{M_1}{}{},...,\encodingm{M_n}{}{}]] \;\contrdiv\; (\res{\ve{b}})(\out{a}{\ve{c}}\para b(z).\encodingm{M_i}{}{}\lrangle{z}) \quad (b\in \ve{b}\subseteq{\ve{c}})
\] The context $D^{x,n}_i$ is defined as ($0<j<i-1; a,b,z$ fresh).
\[
\begin{array}{lcl}
D^{x,n}_i &\DEF& (\res{r_n,p_0',...,p_i',x,x_1,...,x_i,b})([\cdot]\lrangle{r_n} \para \out{a}{x,b} \para \\ 
& & \out{x}{p_0'}.\out{p_0'}{x_1,p_1'}.\cdots.\out{p_{j-1}'}{x_j,p_j'}.\cdots.\out{p_{i-1}'}{x_i,p_i'} \para \\
& & \inp{b}{r_1}.\inp {x_i} {r_2}.(\link{r_2}{r_1}))
\end{array}
\] Then
\DSb 
\[
\begin{array}{rl}
 & D^{x,n}_i[C^{x,n}_{var}[\encodingm{M_1}{}{},...,\encodingm{M_n}{}{}]] \\
\sim  & 
        \begin{array}[t]{l}
(\res{r_n,p_0',...,p_i',x,x_1,...,x_i,b}) (\res{q_0}) \\
( \;\;\;  (\inp{x}{p_0'}.  (\link {p_0'}{q_0}) \para \mathcal{O}\lrangle{q_0,q_n,\encodingm{M_1}{}{},...,\encodingm{M_n}{}{}})   \\ 
\;  \para \out{a}{x,b} \para \out{x}{p_0'}.\out{p_0'}{x_1,p_1'}.\cdots.\out{p_{j-1}'}{x_j,p_j'}.\cdots.\out{p_{i-1}'}{x_i,p_i'}  \\
\; \para \inp{b}{r_1}.\inp {x_i} {r_2}.(\link{r_2}{r_1}))
        \end{array}
 \\
\contrdiv 
&  (\res{r_n,x,x_i,q_i,r_i,b})( \mathcal{O}\lrangle{q_i,q_n,\encodingm{M_{i+1}}{}{},...,\encodingm{M_n}{}{}}) \para  \\ 
 & \out{a}{x,b} \para !\out {x_i} {r_i}.\encodingm{M_i}{}{}\lrangle{r_i} \para  \inp{b}{r_1}.\inp {x_i} {r_2}.(\link{r_2}{r_1})) \\
\sim & (\res{x,x_i,r_i,b})(\out{a}{x,b} \para !\out {x_i} {r_i}.\encodingm{M_i}{}{}\lrangle{r_i} \para  \inp{b}{r_1}.\inp {x_i} {r_2}.(\link{r_2}{r_1})) \\
\sim & (\res{x,x_i,r_i,b})(\out{a}{x,b} \para \inp{b}{r_1}.(!\out {x_i} {r_i}.\encodingm{M_i}{}{}\lrangle{r_i} \para  \inp {x_i} {r_2}.(\link{r_2}{r_1}))) \\
\contrdiv & (\res{x,r_i,b})(\out{a}{x,b} \para \inp{b}{r_1}.(\encodingm{M_i}{}{}\lrangle{r_i} \para \link{r_i}{r_1})) \\
\contrdiv & (\res{x,b})(\out{a}{x,b} \para \inp{b}{r_1}.(\encodingm{M_i}{}{}\lrangle{r_1})) 
\end{array}
\]
\DSe where 
 where the first occurrence of 
$\contrdiv$ subsumes  ($2i+1$)   internal 
$\tau$ actions, and the last step uses property (\ref{scbn-prop-wire}). Thus we are done.
\end{proof}

\section{More properties of the encodings in Section \ref{sec-examples-cbn} and Section \ref{sec-examples-scbn}}\label{ap:property_eg}

\begin{proof}[Proof of Lemma \ref{l:vali}]
For convenience, we recall the content of Lemma \ref{l:vali} and Definition \ref{d:uniCAsyn} below.

\begin{flushleft}{\small
\begin{tabular}{l}
\begin{minipage}{0.97\textwidth}
Lemma \ref{l:vali}. ~Relations $\approx$, $\may$, and $\mayasy$ validate the \uptoE\ technique;
relation $\must$ validates the \uptoEdiv\ technique. 
\end{minipage}
\\\\
\begin{minipage}{0.97\textwidth}
Definition \ref{d:uniCAsyn}. ~ Relation  $\asymp$ {\em validates the \uptoc\ technique} if
for any symmetric relation  $\R$ on $\pi$-processes
we have $\R \subseteq {\asymp}$ whenever
for any pair $(P,Q)\in \R $, if  $P \stm \mu P'$ then  $Q \Arcap \mu
Q'$ and there  are processes $ \til P,  \til Q$ and
a  context  $\qct$  such  that
  $P' \geq  \ct {\til P} $, $Q' \geq  \ct {\til Q}$,
and, if $n\geqNAT 0 $ is the length of the tuples $\til P$ and $\til Q$,
 at least one of the following two statements is true, for each $i \leqNAT n$:
(1) $P_i \asymp Q_i$;
(2) $P_i \RR Q_i$ and,
 if $[\cdot]_i$ occurs under an input in $C$,
also $P_i\sigma \RR  Q_i\sigma$  for all substitutions  $\sigma$.
\end{minipage}
\end{tabular}}
\end{flushleft}
\sepp


As mentioned before, the case for bisimulation is proven in \cite{SW01a}. The cases of may, asynchronous may and must equivalences have similar proofs, which follow from their definitions and use the expansion ($\expa$ for $\may$ and $\mayasy$; $\expa^\Div$ for $\must$) in a way similar to the technique for bisimulation in \cite{SW01a}.
We focus on $\may$ below.
Let
$\R$ be  a relation as in Definition~\ref{d:uniCAsyn}, where $\asymp $ is
$\may$ and $\leq$ is $\expa$. 
 Take the relation
\DSb
\[
{\mathcal{S}} \;\defin\; \asymp \,\cup\, \{(P_1,P_2) \, |\, 
\begin{array}[t]{l}
\mbox{there is a context $\qct$ s.t.} \\
P_i \contr \ct{\til{P_i}}
(i=1,2) \mbox{ and } (\til {P_1}, \til{P_2}) \in \R \,\cup\!\asymp
 \}
\end{array}
 \] 
\DSe
Notation $(\til {P_1}, \til{P_2}) \in \mathcal{R}$ means $(P_1^k, P_2^k) \in \mathcal{R}$
for every $P_i^k\in \til {P_i}$($i=1,2$), $k\leqslant m$ 
\Db
($m$ is the number of holes in
the  $\pi$-context  $C$). \De

Obviously we have $\R \subseteq \mathcal{S}$, then it suffices to show that  $\mathcal{S} \subseteq \asymp$. 
Assume $P_1 \mathcal{S} P_2$. For any context $D$, suppose $D[P_1]\Dwa$.
\Xb
We want to show $D[P_2]\Dwa$.
To do so, we first note that 
from  $P_1 \contr \ct{\til{P_1}}$ we have $D[\ct{\til{P_1}}]\Dwa$ (with not more silent moves before the observable action), 
then we derive $D[\ct{\til{P_2}}]\Dwa$ by a case analysis, and finally we have $D[P_2]\Dwa$ using $\contr$ again. 
We detail the analysis below.
%
\Xe
\Db\finish{the line above is too dense, it would be better to spell out what you want to say}
\De

\DSb
There is a case  analysis to be made  on the origin of the observable in
$D[\ct{\til{P_1}}]\Dwa$, according to where the action in the
observable comes: 
(1)   from $D$; (2)  from $C$; (3)  from $\ve{P_1}$; (4) 
from an interaction between 
\Db
a  component of \De  $\til{P_1}$ and 
its context. 
We only show the details for (3), which is the interesting case;
cases (1) and (2) are easy, and  (4) is easily handled by relying on (1), (2) and (3). 
\DSe

To deal with (3), there are two subcases on $\ve{P_1}\Dwa$: (3{-}1) the observable is from
$P_1^k$ for some $k$; (3{-}2)  the observable is from interaction
between
 components of
 $\ve{P_1}$. We
focus on (3-1) since (3-2) 
can be \Db tackled  similarly to (3-1). \De
For convenience, we set some notations: $\Dwan n$ means  ``observable in $n$ steps'' (of internal move), $\Dwan {\leqslant n}$ means  ``observable in no more than $n$ steps'', 
and $\stm{\tau}_n$ means $n$ consecutive $\tau$ actions.

In the subcase of (3-1), we have $P_1^k\Dwan{n}$ for some $n$, i.e., $P_1^k \stm{\tau}_n\stm{\mu}$ for some $\mu$ other than $\tau$.
We want to show the result $P_2^k\Dwa$ so that the subcase can be closed. 
We know $(P_1^k, P_2^k)\in \R \,\cup\!\asymp$. The 
\Db case when
$(P_1^k, P_2^k)\in\; \asymp$ is immediate.
 \De
For $(P_1^k, P_2^k)\in \R$, we proceed by induction on the $n$ in $P_1^k\!\!\Dwan{n}$ to
show $P_2^k\Dwa$; the property of expansion will be needed.
We elaborate the arguments below. 
\begin{itemize}
\item $n{=}0$. This is trivial based on the definition of $\R$.
\item
 \Db
 Assuming the result holds whenever 
the  number of internal actions before $\mu$
is less than $n$, we show that it also holds for $n{}$. We know that for some $R_1$
\[
P_1^k \stm{\tau} R_1 \stm{\tau}_{n-1} \stm{\mu} \qquad\quad \mbox{ i.e., } R_1\Dwan{n-1}
\] Because $(P_1^k, P_2^k)\in \R$, in terms of Definition \ref{d:uniCAsyn}, we have for some $R_2$
\[
P_2^k \wt{\tau} R_2 
\] and for some context $C',\ve{R_1},\ve{R_2}$ such that $(\ve{R_1},
\ve{R_2})\in \R\cup\!\asymp$, it holds that
\[
R_1 \geq C'[\ve{R_1}] \quad\mbox{ and }\quad C'[\ve{R_2}]\leq R_2
\] Since $R_1\Dwan{n-1}$, we have
\[
C'[\ve{R_1}]\Dwan{\leqslant {n-1}}
\] Then by (possibly) using induction hypothesis, we derive
\[
 C'[\ve{R_2}]\Dwa
\] Thus
\[
R_2\Dwa
\] So in summary
\[
R_1(\Dwan{n-1}) \geq C'[\ve{R_1}](\Dwan{\leqslant {n-1}}) \quad\mbox{ and }\quad C'[\ve{R_2}](\Dwa) \leq R_2(\Dwa)
\] Hence we finally have
\[
P_2^k\Dwa \tag*{\qEd}
\]
\De
\end{itemize}
\def\popQED{}
\end{proof}

\section{Discussion of Section \ref{s:typ_asy}}\label{ap:property_eg_type}

We briefly introduce the asynchronous $\pi$-calculus with linear types, based on the calculus in Section \ref{sec-background}. The reader is referred to \cite{SW01a} for more details. 
After that, we explain how to adapt the conditions to a setting allowing types, and prove Theorem \ref{t:resCBNlin} from Section \ref{s:typ_asy}.

\xx{ TODO (refer to "sourcecodeEAST" and  \cite{SW01a,YoshidaHB07})
\begin{itemize}
\item[-(1)]  \xx{ Formulate the linearly typed asynchronous pi-calculus (1. plain types; 2. io (sub) types; 3. linear types.);}
\item[-(2)]  \xx{ Define the (typed) asynchronous barbed congruence;}
\item[-(3)]  \xx{ Prove Theorem \ref{t:resCBNlin};}
\item[-(4)]  \xx{ Adjust (mutate) the conditions under types (in terms of clues in the proof);}
\item[-(5)]  \xx{ Check everything \& tidy up.} \\
\rc{Structure:}
\begin{itemize}
\item \xx{the asynchronous $\pi$-calculus with linear types: syntax, semantics, barbed congruence and its properties;}
\item \xx{the conditions: reuse the conditions for BT in untyped case \& present Figure \ref{f:condtions:types} (say that types are used 'implicitly' in the conditions, e.g., in $\asymp$);}
\item \xx{the linearly typed encoding of Milner;}
\item \xx{the proof of Theorem \ref{t:resCBNlin};}
\item \xx{(if any) discussion.}
\end{itemize}
\rc{\large (in several rounds)} \xx{reorganize this appendix from the next (clear) page}
\item[-$>$(6)] \xx{Communicate \& decorate.}
\end{itemize}
}



\subsection{Linearity: types, typing (rules), barbed congruence, and bisimulation}
\xx{\scriptsize
\begin{itemize}
\item[Note:] based on LTS in \cite{SW01a}, only soundness w.r.t. barbed congruence holds, i.e. (labelled) bisimilarity implies barbed congruence; completeness is not true \cite{SW01a}; of course this may be relevant to the design of LTS
\item[Question:] \rc{Does the context lemma hold in presence of linear types? See Lemma \ref{l:context}}
\end{itemize}
}
\sepp

\subsubsection{Asynchronous $\pi$-calculus with linear types}
The linearly typed asynchronous $\pi$-calculus, notation $\pi^l$, is defined in \afig
\ref{f:linear_types_lapi} (types), \afig \ref{f:types_operations} (operation on types),
\afig \ref{f:syntax_lapi} (syntax), \afig \ref{f:typing_rules_lapi} (typing), and \afig
\ref{f:trans_rule_lapi} (semantics). They are based on the corresponding part in
\Db
 \cite[Chapter 8]{SW01a}.
\De
We start with some remark about the notations.
\begin{itemize}
\item Notation $\ve{b} : \ve{T}$ is a shortcut for $b_i : T_i$ ($i=1,...,n$ where $n$ is the size of $\ve{b}$).
\item 
$\TPEQ$ is type equality, defined as the (smallest) congruence satisfying the rule below.
\begin{mathpar}
\inferrule*[left=EQ-UNFOLD]{ }{\mu X.T \;\TPEQ\; T\hosub{\mu X.T}{X}}
\end{mathpar}
\end{itemize}
\Db 
The figures \ref{f:linear_types_lapi}-\ref{f:trans_rule_lapi}  follow the formulation in
\cite{SW01a}, to which we refer for more details.

The following are standard definitions and notations concerning type environments.
\begin{definition}[Type environment]
$ $ 
\begin{itemize}
\item An \emph{assignment} is of the form $a:T$, meaning that $a$ gets type $T$.
\item A \emph{type environment}, represented by $\Gamma, \Delta$, is a finite set of assignments.
\item Metavariable $\Theta$ ranges over type environment.
\item Given a type environment $\Gamma$, $\Gamma(a)$ stands for the type assigned to $a$ by $\Gamma$, and $\supp{\Gamma} \DEF \{a \mid \Gamma(a) \mbox{ is defined}\}$.
\item $\Gamma\backs{a}$ is the type environment excluding only the definition on $a$ in $\Gamma$.
\item {A type environment is \emph{closed} if it does not contain free type variables in its assignments, and $\Gamma(a)$ is a link type for all $a\in \supp{\Gamma}$.} By default, we consider closed type environments. 
\xx{(Maybe no harm to drop the extra requirement like that in Def. 8.1.1 of \cite{SW01a}(P.285))}
\item \emph{$\Gamma$ extends $\Delta$} if $\supp{\Delta} \subseteq \supp{\Gamma}$ and
  $\Gamma \vdash x:\Delta(x)$ for every $x\in \supp{\Delta}$. 
\end{itemize}
\end{definition}
\sepp
\De

\begin{figure}[t]
\noindent\rule{\textwidth}{.5pt}\\
\begin{minipage}{11.6cm}

\vspace{0.5em}
~~~~{Types}
\[
\begin{array}{rcll}
S,T \qquad & ::= \qquad & L \;\; & \mbox{link type} \\
    & | \qquad & \Diamond & \mbox{behavior type}\\
    & | \qquad & L\rightarrow\Diamond & \mbox{abstraction type}\\\\\\
L \qquad & ::= \qquad & \mathbf{unit}  & \mbox{basic type} \\
   & | \qquad & \sharp L & \mbox{connection} \\
   & | \qquad & i L & \mbox{input capability} \\
   & | \qquad & o L & \mbox{output capability} \\
   & | \qquad & l_\sharp L & \mbox{linear connection} \\
   & | \qquad & l_i L & \mbox{linear input capability} \\
   & | \qquad & l_o L & \mbox{linear output capability} \\
   & | \qquad & \prod_{i=1}^{n}L_i \; (n\geqslant 2)  & \mbox{product} \\
   &  \qquad & \mbox{(briefly $\seq{L}$) } & \mbox{} \\
   & | \qquad & X & \mbox{{type variable}} \\
   & | \qquad & \mu X.L & \mbox{{recursive type}} \\\\\\
\mbox{Type environments} \qquad  & & & \\\\
\Gamma \qquad & ::= \qquad & \Gamma, x : L \qquad & \\
   & | \qquad & \Gamma, \seq{x} : \seq{L} & \\
   & | \qquad & \emptyset &
\end{array}
\]
\vspace{0.5em}
\end{minipage}

\noindent\rule{\textwidth}{.5pt}\caption{Types (including linear types) }\label{f:linear_types_lapi}
\end{figure}

\begin{figure}[htbp]
\noindent\rule{\textwidth}{.5pt}
\begin{minipage}{11.6cm}

\vspace{0.5em}
{Combination of types}
\[
\begin{array}{rcll}
l_i T \uplus l_o T & \;\;\DEF\;\; & l_\sharp T \quad &  \\
T \uplus T & \;\;\DEF\;\; & T \quad & \mbox{ if $T$ is not a linear type } \\
S \uplus T & \;\;\DEF\;\; & \mathbf{error} \quad & \mbox{ otherwise } \\
\end{array}
\]

{Combination of type environments}
\[
(\Gamma_1\uplus \Gamma_2)(x) \DEF \left\{
\begin{array}{ll}
\Gamma_1(x) \uplus \Gamma_2(x) &  \mbox{ if both $\Gamma_1(x)$ and $\Gamma_2(x)$ are defined } \\
\Gamma_1(x) &  \mbox{ if $\Gamma_1(x)$ is defined, but not $\Gamma_2(x)$} \\
\Gamma_2(x) &  \mbox{ if $\Gamma_2(x)$ is defined, but not $\Gamma_1(x)$} \\
\mbox{undefined} & \mbox{ if neither $\Gamma_1(x)$ nor $\Gamma_2(x)$ is defined, } \\
  & \mbox{ or both are defined but $\Gamma_1(x) \uplus \Gamma_2(x){=}\mathbf{error}$}
\end{array}\right.
\]

{Extraction of linear names}
\[
\begin{array}{rcl}
Lin(\Gamma) & \DEF & \{x \;\mid\; \Gamma(x) \mbox{ is a linear type}\}   \\
Lin_i(\Gamma) & \DEF & \{x \;\mid\; \Gamma(x)=l_iS \mbox{ or } \Gamma(x)=l_\sharp S, \mbox{ for some } S\}
\end{array}
\]
\vspace{0.5em}
\end{minipage}
\noindent\rule{\textwidth}{.5pt}

\caption{Operations on types}\label{f:types_operations}
\end{figure}

\begin{figure}[t]
\noindent\rule{\textwidth}{.5pt}
\begin{minipage}{11.6cm}

\[
\begin{array}{ccll}
  & & & \mbox{Values} \\
v \qquad & ::= \qquad & x \qquad & \mbox{name} \\
    & | \qquad & \seq{x} & \mbox{name product}\\
    & | \qquad & \star & \mbox{basic value}\\\\\\
  & & & \mbox{Agents} \\
A \qquad & ::= \qquad & P \qquad & \mbox{process} \\
    & | \qquad & F & \mbox{abstraction}\\\\\\
  & & & \mbox{Processes} \\
P,Q,R \qquad & ::= \qquad & \nil \qquad\qquad\qquad & \mbox{null} \\
    & | \qquad & \inp a {\seq{b}} . P & \mbox{input}\\
    & | \qquad & \out a {\seq{b}}  & \mbox{output}\\
    & | \qquad & P | Q & \mbox{parallel composition}\\
    & | \qquad & (\res {a:L}) P & \mbox{restriction}\\
    & | \qquad & ! \inp a {\seq{b}} . P & \mbox{replication}\\
    & | \qquad & \app F a & \mbox{application}\\
    & | \qquad & \mathbf{wrong} & \mbox{error}\\\\\\
  & & & \mbox{Abstractions} \\
F \qquad & ::= \qquad & \abs {a} P \qquad & \mbox{} \\

\end{array}
\]
\vspace{0.5em}
\end{minipage}
\noindent\rule{\textwidth}{.5pt}

\caption{Syntax}\label{f:syntax_lapi}
\end{figure}

\begin{figure}[htbp]
\noindent\rule{\textwidth}{.5pt}
{\small
\begin{minipage}{11.6cm}

\vspace{0.5em}
{Value typing}\\
\begin{mathpar}
\inferrule[TV-BASE]{Lin(\Gamma)=\emptyset}{\Gamma  \vdash \star : \mathbf{unit}}    \and
\inferrule[TV-NAME]{Lin(\Gamma)=\emptyset}{\Gamma, x : T \vdash x : T}  \and
\inferrule[TV-PRODUCT]{\Gamma_i \vdash x_i : T_i \and i=0,1,....,n}{\biguplus_i{\Gamma_i} \vdash \seq{x} : \seq{T}}  \and
\inferrule[TV-SUBSUMPTION]{\Gamma \vdash x : S \and S \subtp T}{\Gamma \vdash x : T}   \and
\inferrule[TV-EQ]{\Gamma \vdash x : S \and S \;\TPEQ\; T}{\Gamma \vdash x : T}
\end{mathpar}

{Subtyping}
\begin{mathpar}
\inferrule*[left=SUB-REFL]{ }{T \subtp T}    \and
\inferrule*[left=SUB-TRANS]{S \subtp S'\and S'\subtp T }{S \subtp T}    \and
\inferrule*[left=SUB-$\sharp$I]{ }{\sharp T \subtp  iT}    \and
\inferrule*[left=SUB-$\sharp$O]{ }{\sharp T \subtp  oT}    \and
\inferrule*[left=SUB-II]{S \subtp T}{iS \subtp iT}    \and
\inferrule*[left=SUB-OO]{T \subtp S}{oS \subtp oT}    \and
\inferrule*[left=SUB-PRODUCT]{S_i \subtp T_i \and i=0,1,...,n}{\seq{S} \subtp \seq{T}}    \and
\inferrule*[left=SUB-BB]{S \subtp T \and T \subtp S}{\sharp S \subtp \sharp T}    \and
\end{mathpar}

{Process typing}
\begin{mathpar}
\inferrule*[left=T-NIL]{Lin(\Gamma)=\emptyset}{\Gamma \vdash \nil : \Diamond}    \and
\inferrule*[left=T-$\Delta$-HOLE]{\Theta \mbox{ extends the (fixed) } \Delta }{\Theta \vdash \holem : \Diamond}    \and
\inferrule*[left=T-INP]{\Gamma_1 \vdash a : \varsigma \seq{T} \;(\varsigma\in \{i,l_i\}) \and \Gamma_2, \seq{b} : \seq{T} \vdash P : \Diamond}{\Gamma_1\uplus\Gamma_2 \vdash \inp  a {\seq{b}}.P : \Diamond}    \and
\inferrule*[left=T-OUT]{\Gamma_1 \vdash a: \varsigma \seq{T} \;(\varsigma\in \{o,l_o\})\and  \Gamma_2 \vdash \seq{b} : \seq{T}}{\Gamma_1\uplus\Gamma_2 \vdash \out a {\seq{b}} : \Diamond}  \and
\inferrule*[left=T-PAR]{\Gamma_1 \vdash P : \Diamond \and  \Gamma_2 \vdash Q : \Diamond}{\Gamma_1\uplus\Gamma_2 \vdash P \para Q : \Diamond}    \and
\inferrule*[left=T-RES]{\Gamma, a : L \vdash P : \Diamond \and }{\Gamma \vdash (\res {a : L}) P : \Diamond}   \and 
\inferrule*[left=T-RES']{\Gamma\vdash P : \Diamond \and }{\Gamma \vdash (\res {a : L}) P : \Diamond}   \and 
\inferrule*[left=T-REP]{\Gamma_1 \vdash a : i \seq{T} \and \Gamma_2, \seq{b} : \seq{T} \vdash P : \Diamond \and Lin(\Gamma_2)=\emptyset}{\Gamma_1\uplus\Gamma_2 \vdash !\inp  a {\seq{b}}.P : \Diamond}   \and
\inferrule*[left=T-ABS]{\Gamma, a : T \vdash P : \Diamond}{\Gamma \vdash \abs {a} P : T\rightarrow \Diamond} \and
\inferrule*[left=T-APP]{\Gamma_1 \vdash  F : T\rightarrow \Diamond \and \Gamma_2 \vdash b : T}{\Gamma_1\uplus\Gamma_2 \vdash \app F b : \Diamond} \enspace.
\end{mathpar}
\vspace{0.5em}
\end{minipage}
}
\noindent\rule{\textwidth}{.5pt}

\caption{Typing rules}\label{f:typing_rules_lapi}
\end{figure}

\begin{figure}[htbp]
\noindent\rule{\textwidth}{.5pt}
\begin{minipage}{11.6cm}
\vspace{0.5em}

\begin{center}
\begin{tabular}{rlrl}
{\trans{ inp}}:& $ \inp  a {\seq{b}} . P \stm{\inp a {\seq{b}}} P$  &
\trans{ rep}:& $  !\inp a {\seq{b}} . P   \stm{ \inp a {\seq{b}}}  P | ! \inp a {\seq{b}} . P  $, if $a \not \in {\seq{b}}$ \\[\mysp]
{\trans{ out}}:& $ \out a {\seq{b}}     \stm{\out a {\seq{b}}} \nil $  &
{\trans{ par}}:& $\displaystyle{   P \stm\mu   P' \over   P | Q   \stm\mu P'| Q } $ if $\bn \mu \cap \fn Q = \emptyset $   \\[\mysp]
\multicolumn{4}{c}{
  {\trans{com}}: $ \;\;    \displaystyle{ P \stm{\inp  a{\seq{c}} }P' \hskip .4cm   Q \stm{(\res {\til{d} : \til{L}}) \out a {\seq{b}}} Q'  \over     P  |  Q \stm{ \tau} (\res{\til{d} : \til{L}})( P' \sub {\seq{b}} {\seq{c}} |  Q' )}$ if  $\til{d} \cap \fn P = \emptyset $
} \\[\mysp]
\multicolumn{4}{c}{
{\trans{ res}}: \; \; $\displaystyle{ P \stm{\mu}P' \over  (\res {a : L})     P   \stm{ \mu} (\res {a : L}) P'  } $ $ a$ does not appear in $\mu$ } \\[\mysp]
\multicolumn{4}{c}{
{\trans{ open}}:\;\; $\displaystyle{ P \stm{(\res{\til{d} : \til{L}}) \out a {\seq{b}} }P' \over  (\res {c:L})     P   \stm{ (\res{c:L,\til{d}:\til{L}}) \out a {\seq{b}}  }  P'  } $ $c\in \tilb -\til{d}, \;  a \neq  c$.
} \\[\mysp]
\multicolumn{4}{c}{
 {\trans{ app}}:  $ \; \; \displaystyle{ P\sub b a \stm{\mu}P' \over  \app F b   \stm{ \mu}  P'  } $ if  $F = \abs {a} P$} \\[\mysp]
\multicolumn{4}{c}{
{\trans{ inpErr}}: $ \inp  a {\seq{b}} . P \stm{\tau} \mathbf{wrong}$ \; (\mbox{$a$ is not a name})
} \\[\mysp]
\multicolumn{4}{c}{
\trans{ outErr}: $ \out a {\seq{b}}     \stm{\tau} \mathbf{wrong} $ \; (\mbox{$a$ is not a name})
}\\[\mysp]
\end{tabular}
\end{center}
\end{minipage}
\noindent\rule{\textwidth}{.5pt}

\caption{Transition rules}\label{f:trans_rule_lapi}
\end{figure}

\subsubsection{Asynchronous typed barbed congruence}
\

We give the definition of barbed congruence in $\pi^l$, and some of its properties.
The following notion is used to express the quantification over contexts in the definition of barbed congruence.
\Db
\begin{definition}[($\Gamma/\Delta$)-context] 
Context $C$ is a \emph{$(\Gamma/\Delta$)-context} if, assuming $\holem$ as a process, either the judgement $\Gamma \vdash C:\Diamond$ or $\Gamma \vdash C:T\rightarrow\Diamond$ is valid.
\end{definition}
Note that if $C$ is a $(\Gamma/\Delta$)-context and $\Delta \vdash P$, then $\Gamma \vdash
C[P]$.
\De

\Db
Notation $P\da_{\mu}$ (respectively $P\Da_{\mu}$) means $P\stm{\mu}$ (respectively $P\wt{}\stm{\mu}$).
\begin{definition}[Asynchronous typed barbed bisimilarity and congruence]
\

\begin{enumerate}
\item \emph{Asynchronous typed barbed bisimilarity} is the largest symmetric relation, $\BSTA$, such that whenever $P \BSTA Q$ \\
(a) [Barb-preserving] $P\da_{\mu}$ implies $Q\Da_{\mu}$, where $\mu$ is an output; \\
(b) [Reduction-closed] $P\stm{\tau} P'$ implies $Q\wt{} \BSTA P'$.
\item Suppose $\Delta\vdash P,Q$ for some $\Delta$. Processes $P$ and $Q$ are 
\emph{asynchronously typed barbed congruent (w.r.t. $\Delta$)}, written $\Delta \Vdash P \BCTA Q$, if for every $(\Gamma/\Delta$)-context $C$ (in which $\Gamma$ is closed), one has $C[P] \BSTA C[Q]$.
\end{enumerate}
\end{definition}
\De
Some technical concepts and results are given below without comments. They are
\Db discussed and proved in \cite{SW01a}. We 
\De
refer the reader to \cite{SW01a} for more details. 

\begin{definition}[$\Delta$-to-$\Gamma$ substitution \xx{(ref.\cite{SW01a}:P.324)}\!\!]
Suppose $\Delta,\Gamma$ are type environments, and $\sigma$ is a substitution on names. We
say that $\sigma$ is a 
\Db \emph{$\Delta$-to-$\Gamma$ substitution} if
\De
\begin{enumerate}
\item[] for every $x$ on which $\Delta$ is defined, $\Gamma\vdash \sigma(x) : \Delta(x)$.
\end{enumerate}
\end{definition}

\begin{lemma}[Substitution Lemma \xx{(ref.\cite{SW01a}:P.285)}\!\!] \label{l:sub-lin}
Assume $\Gamma,a:S\vdash A:T$, $\Gamma'\vdash b:S$, and $\Gamma\uplus\Gamma'$ is
defined. Then $\Gamma\uplus\Gamma' \vdash A\hosub{b}{a}:T$. 
\Db \De
\end{lemma}

\begin{lemma}[Subject Reduction Lemma \xx{(ref.\cite{SW01a}:P.286)}\!\!]\label{l:subj-red--lin}
Let $\Gamma$ be closed and $\Gamma \vdash P$. Suppose $P\stm{\alpha} P'$.
\begin{enumerate}
\item If $\alpha$ is $\tau$, then either $\Gamma \vdash P'$ or there exist $a,T$ with $\Gamma(a)=l_\sharp T$ such that $\Gamma\backs{a} \vdash P'$.
\item If $\alpha$ is $a(\seq{c})$, then there exist $\Gamma_1,\Gamma_1,\seq{S}$ such that
\begin{enumerate}
\item $\Gamma = \Gamma_1\uplus\Gamma_2$;
\item $\Gamma_1 \vdash a:i\seq{S}$ or  $\Gamma_1 \vdash a:l_i \seq{S}$;
\item if $\Gamma_3\vdash \seq{c}:\seq{S}$ and $\Gamma_2\uplus\Gamma_3$ is defined, then $\Gamma_2\uplus\Gamma_3 \vdash P'$.
\end{enumerate}
\item If $\alpha$ is $(\res{\til{d} : \til{T}})\out {a} {\seq{c}}$, then there exist $\Gamma_1,\Gamma_2,\Gamma_3,\seq{S}$ such that
\begin{enumerate}
\item $\Gamma,\til{d}:\til{T} = \Gamma_1\uplus\Gamma_2\uplus\Gamma_3$;
\item $\Gamma_1 \vdash a:o\seq{S}$ or  $\Gamma_1 \vdash a:l_o \seq{S}$;
\item $\Gamma_2\vdash \seq{c}:\seq{S}$;
\item $\Gamma_3\vdash P'$.
\end{enumerate}
\end{enumerate}
\end{lemma}
\sep

Lemma \ref{l:context} offers a characterization of $\BCTA$, as an extension to related results in \cite{SW01a}.

\begin{definition}
\label{d:context_def}
Suppose $\Delta\vdash P,Q$. We \Db write \De
 $\Delta \Vdash P \CTTA Q$ if, for every closed $\Gamma$ that extends $\Delta$, every $\Delta$-to-$\Gamma$ substitution $\sigma$, and every process $R$ such that $\Gamma\vdash R$, we have $P\sigma \para R \BSTA Q\sigma \para R$.
\end{definition}

\DSb
The following Context Lemma is useful when reasoning about the
behaviour of processes that are barbed congruent,
so to drastically limit the quantification on contexts by appealing to
the above Definition~\ref{d:context_def}. In this way, the reasoning
can become similar to that employed when working with ordinary labeled
bisimilarity   in the untyped case (e.g., case (4) in the proof of
 Theorem \ref{t:resCBNlin} in Section~\ref{sss:resCBNlin};
see also
 \cite{SW01a} for other examples and discussion).
\DSe

\begin{lemma}[Context Lemma for linearity]
\label{l:context}
Suppose $\Delta\vdash P,Q$. \\
{It holds that ~ $\Delta \Vdash P \BCTA Q$ ~iff~ $\Delta \Vdash P \CTTA Q$.}
\end{lemma}
\begin{proof}
\xx{(ref.\cite{SW01a}:P.342)}

The implication that $\BCTA$ implies $\CTTA$ is easy. For the other direction, we prove by induction on the structure of ($\Gamma/\Delta$)-context $C$ (in which $\Gamma$ is closed) that
\[
\Delta \Vdash P \CTTA Q \quad\mbox{ implies }\quad \Gamma \Vdash C[P] \CTTA C[Q]
\]

We first give two claims whose proofs are similar to those for the untyped $\pi$-calculus. They are used in the analysis of this lemma.
\begin{itemize}
\item[] {Claim 1.} If $\Delta \Vdash P \CTTA Q$ and $\Gamma$ extends $\Delta$, then $\Gamma \Vdash P \CTTA Q$.
\item[] {Claim 2.} If $\Delta \Vdash P \CTTA Q$ and $\Delta(a){=}S$, then $\Delta\backs{a} \Vdash (\res {a:S})P \CTTA (\res {a:S})Q$.
\end{itemize}

To proceed, there are a couple of cases to analyze.
\begin{itemize}
\item $C$ is $\nil$ or $\out a {\seq{b}}$. This is trivial.
\item $C$ is $\holem$. This is by Claim 1.
\item $C$ is $R\para C'$. This is immediate by induction hypothesis and the premise.
\item $C$ is $(\res {a:S}) C'$. This is by Claim 2.
\item $C$ is $!\inp a {\seq{b}}.C'$. This case is similar to that for i/o types; see \cite{SW01a}.

\item $C$ is $\inp a {\seq{b}}.C'$.
We focus on the subcase when $a$ is of a linear type. Otherwise the argument is similar to that for i/o types in \cite{SW01a}.
Given $C$ as a ($\Gamma/\Delta$)-context, we have
\begin{equation}\label{contextlemma-eq1}
\Gamma \vdash a : l_i \seq{S} \quad\mbox{ and }\quad C' \mbox{ is a ($(\Gamma,\seq{b}:\seq{S})/\Delta$)-context}
\end{equation}
The aim is to prove that for every closed $\Gamma'$ extending $\Gamma$, every $R$ such that $\Gamma'\vdash R$, and every $\Gamma$-to-$\Gamma'$ substitution $\sigma$,  the relation $\R \cup \BSTA$ is a barbed bisimulation, where $\R$ is defined as
\begin{equation}\label{contextlemma-eq2}
\Big\{\big( (\inp a {\seq{b}}.C'[P])\sigma \para R ,  (\inp a {\seq{b}}.C'[Q])\sigma \para R \big) \;\Big|\; \sigma,R \mbox{\small ~ are as described above}\Big\}
\end{equation}
Then it holds that
\begin{equation*}
(\inp a {\seq{b}}.C'[P])\sigma \para\! R \;\BSTA\;  (\inp a {\seq{b}}.C'[Q])\sigma \para\! R
\end{equation*}
Take an element from $\R$, the barb-preserving property should be clear since no immediate
output can be made by $(\inp a {\seq{b}}.C'[P])\sigma$ or $(\inp a
{\seq{b}}.C'[Q])\sigma$. We thus, in the
\DSb
remainder
\DSe of the proof, consider the reduction-closed requirement, and show that every reduction of, say, $(\inp a {\seq{b}}.C'[P])\sigma \para\!\! R$ can be matched by $(\inp a {\seq{b}}.C'[Q])\sigma \para\!\! R$. The most interesting case is when the reduction results from the interaction between $(\inp a {\seq{b}}.C'[P])\sigma$ and $R$. That is,
\begin{equation}\label{contextlemma-eq3}
(\inp a {\seq{b}}.C'[P])\sigma \para\! R \;\;\stm{\tau}\;\; (\res{\til{d} : \til{T}})(C'[P]\sigma\fosub{\seq{c}}{\seq{b}} \para\! R')
\end{equation}
The reduction is shown below to be matched by
\begin{equation}\label{contextlemma-eq4}
(\inp a {\seq{b}}.C'[Q])\sigma \para\! R \;\;\stm{\tau}\;\; (\res{\til{d} : \til{T}})(C'[Q]\sigma\fosub{\seq{c}}{\seq{b}} \para\! R')
\end{equation}
That is,
\begin{equation}\label{contextlemma-eq5}
(\res{\til{d} : \til{T}})(C'[P]\sigma\fosub{\seq{c}}{\seq{b}} \para\! R')  \;\BSTA\; (\res{\til{d} : \til{T}})(C'[Q]\sigma\fosub{\seq{c}}{\seq{b}} \para\! R')
\end{equation}
which is derivable if we can prove the \Db  following, because $\BSTA$ is closed by
restriction: \De
\begin{equation}\label{contextlemma-eq6}
C'[P]\sigma\fosub{\seq{c}}{\seq{b}} \para\! R'  \;\BSTA\; C'[Q]\sigma\fosub{\seq{c}}{\seq{b}} \para\! R'
\end{equation}
To this end, (\ref{contextlemma-eq6}) can be inferred by induction hypothesis 
\Db
on $C'$, 
which is  
a ($(\Gamma,\seq{b}:\seq{S})/\Delta$)-context, if  one \De can exhibit that for closed $\Gamma''\DEF \Gamma',\til{d}:\til{T}$ it holds that
\begin{eqnarray}
\sigma\fosub{\seq{c}}{\seq{b}} \mbox{ is a $(\Gamma,\seq{b}:\seq{S})$-to-$\Gamma''$ substitution} \label{contextlemma-eq7} \\
\Gamma'' \vdash R' \label{contextlemma-eq8}
\end{eqnarray}
In (\ref{contextlemma-eq3}), the reduction results from an output from $R$, that is,
\[
R \stm{(\res{\til{d} : \til{T}})\out {a} {\seq{c}}} R'
\]
The most intriguing situation here is when $R$ has the linear output capability on $a$, i.e., $\Gamma'\vdash a:l_o T$.
As $\Gamma' \vdash R$, by the Subject Reduction lemma (Lemma \ref{l:subj-red--lin}), we have, for some $\Gamma_1,\Gamma_2,\Gamma_3$
\begin{eqnarray}
\Gamma',\til{d}:\til{T} = \Gamma_1\uplus\Gamma_2\uplus\Gamma_3 \label{contextlemma-eq9}\\
\Gamma_1 \vdash a:l_o \seq{S} \label{contextlemma-eq10}\\
\Gamma_2\vdash \seq{c}:\seq{S} \label{contextlemma-eq11}\\
\Gamma_3\vdash R'\label{contextlemma-eq12}
\end{eqnarray}
From (\ref{contextlemma-eq9}) and (\ref{contextlemma-eq12}), we infer (\ref{contextlemma-eq8}).
Then (\ref{contextlemma-eq7}) can be derived by (\ref{contextlemma-eq13}) below and (\ref{contextlemma-eq10}).
\begin{eqnarray}
\Gamma',\til{d}:\til{T} \vdash \sigma(y):\Gamma(y),  \mbox{ for every $y$ defined in $\Gamma$} \label{contextlemma-eq13}
\end{eqnarray}
Moreover, (\ref{contextlemma-eq13}) is valid because $\sigma$ is a $\Gamma$-to-$\Gamma'$ substitution.
This completes the case and the proof.\qedhere
\end{itemize}
\end{proof}

\ifxxremark
\noindent\emph{Remark}.~ [Some further direction? consider \underline{Moving to Section Conclusion} (temp. keep here)]
\begin{itemize}
\item \ovalbox{{\small What if allow subtyping on linearity, i.e., something like $\infer{l_iS\subtp l_iT}{S\subtp T}$ ~ ??}}
\item \ovalbox{{\small What if also allow combination of non-linear type, i.e., something like $iT\uplus oT \DEF \sharp T$, overriding operation $\uplus$:}}
{\small
\[
\begin{array}{rcll}
l_i T \uplus l_o T & \;\;\DEF\;\; & l_\sharp T \quad &  \\
T \uplus T & \;\;\DEF\;\; & T \quad & \mbox{ if $T$ is not a linear type } \\
iT \uplus oT & \;\;\DEF\;\; & \sharp T \\
\sharp T \uplus oT & \;\;\DEF\;\; & \sharp T \\
iT \uplus \sharp T & \;\;\DEF\;\; & \sharp T \\
S \uplus T & \;\;\DEF\;\; & \mathbf{error} \quad & \mbox{ otherwise } \\
\end{array}
\]}
\rc{Does it make sense; is it worth some new development on linearity? Maybe not in theory, but yes in practice.}
\end{itemize}
\fi

\subsection{Proof for Section \ref{s:typ_asy}}
\xx{First, Adjust by adding type environment to the conditions...and }
\xx{Show the following theorem}

\subsubsection{The conditions for full abstraction w.r.t. BT}
\

We reuse the conditions for BT in untyped case (Section \ref{ss:sounds_completes}).
To adapt to the case for typed $\pi$, we assume that types are used 'implicitly' in the conditions (e.g., in $\asymp$), so as to maintain succinctness.
For convenience, we reproduce the conditions in \afig \ref{f:condtions:types} for use shortly.


\begin{theorem}
Let $\qenco$ be an encoding of the $\lambda$-calculus into $\pi$-calculus with linear types, $\asymp$ \xx{(should it be parametric on a type environment $\Delta$?` seems so; see blue in this thm!)} a congruence on $\pi$-agents.
Suppose  there are a  precongruence  $\leq$ on $\pi$-agents and \xxstress{a type $T_b$
  assigned to the abstracted names of the encoding  such that $\asymp$ is constrained by a
  type environment respecting $T_b$ (i.e., typing the abstracted names of the encoding
  with $T_b$, and every \Db term $\encoding{M}{}{}$ \De 
has type $T_b\rightarrow \Diamond$)}.
If the conditions in \afig \ref{f:condtions:types} hold,  then $\qenco$ and $\asymp$ are fully abstract for BTs.
\end{theorem}
\begin{proof}
Types stipulate the shape of a process (including the encoding of a $\lambda$ term), and
do not play a part in reductions, so the proof is conducted in a way \Db  similar \De to the case without types.

In the completeness proof, 
\Db the important part is the one about the \uptoc\ technique, whereas in the soundness
proof
the important part is the one about the context-inverse properties. 
The proofs for these  parts do not change with respect to the untyped case,
 since we use the same expansion relation $\expa$ as before. \qedhere
\De

\xx{ This is it ? ok`?`}
\end{proof}

\begin{figure}[t] 
\begin{center}
\noindent\rule{\textwidth}{.5pt}
\begin{tabular}{l}
\begin{minipage}{12cm}
\vspace{0.5em}
Let $\asymp$ and $\leq$ be relations on agents of asynchronous $\pi$-calculus with linear types. \\
\xxstress{We assume}, in addition to the untyped version, that $\asymp$ is subject to type environment that respects $T_b$, that is, the ``address" of the encoding of a $\lambda$ term is assigned this type and every encoding has type $T_b\rightarrow\Diamond$.
\end{minipage}
\\\\

\begin{minipage}{12cm}\begin{center}Completeness conditions for BT\end{center}\vspace*{-.1cm}
\begin{enumerate}
\item \label{i:sLLa:tp} 
$\asymp$ is a congruence and  $  {\asymp} \supseteq{ \geq}$; ~~~~~~~~ \xx{\rc{\large as before $\checkmark$}}

\item \label{i:expa:tp} 
$\leq$ is an expansion relation and is a
plain
 precongruence; ~~~~~~~~ \xx{\rc{\large as before $\checkmark$}}

\item \label{i:sLLf:tp} 
$\asymp$ validates the \uptoc\ technique; ~ \xx{\rc{\large jump to (Claim 4) $\checkmark$}} 

\item \label{i:sLLb:tp} 
the  variable  contexts of $\qenco$ are guarded; ~~~~~~~~ \xx{\rc{\large as before $\checkmark$}}

\item \label{i:sLLg:tp} 
 $\qenco$ and $\geq$ validate rule $ \beta$; ~~~~~~~~ \xx{\rc{\large as before $\checkmark$}}

\item \label{i:sLLh:tp} 
if  $M $ is an unsolvable of order $0$ then $\encom M \asymp \encom \Omega$; ~~~~~~~~ \xx{\rc{\large Claim 3 $\checkmark$}}

\item \label{i:sLLk:tp} 
$ \encom {M} \asymp \encom{\Omega} $ whenever $M$ is unsolvable of order $\infty$. ~~~~~~~~ \xx{\rc{\large Claim 2 $\checkmark$}}
\end{enumerate}
\end{minipage}
\\\\

\begin{minipage}{12cm}\begin{center}Soundness conditions for BT\end{center}\vspace*{-.1cm}
\begin{enumerate}
\item \label{con:lls:a:tp} 
$\asymp$  is a congruence, $\leq$  a
plain
 precongruence; \xx{\rc{\large as before $\checkmark$}}

\item \label{con:lls:b:tp} 
$ { \asymp} \supseteq {\geq}$;  \xx{\rc{\large as before $\checkmark$}}

\item \label{con:lls:c:tp} 
$ { \asymp}$ has the rendez-vous cancellation property; ~ \xx{\rc{\large jump to (Claim 5) $\checkmark$}}

\item \label{con:lls:d:tp} 
$\qenco$ and $\geq$ validate rule $\beta$;  ~~~~~~~~ \xx{\rc{\large as before $\checkmark$}}

\item \label{con:lls:e:tp} 
if $M$ is an unsolvable of order $0$, then $\encodingm{M}{}{} \asymp \encodingm{\Omega}{}{} $; ~~~~~~~~ \xx{\rc{\large Claim 3 $\checkmark$}}

\item \label{con:lls:f:tp}
the terms $\encodingm{\Omega}{}{}$, $\encodingm{x\ve M}{}{}$, $\encodingm{x\ve {M'}}{}{}$, and $\encodingm{y\ve {M''}}{}{}$ are pairwise unrelated by $\asymp$, assuming that $x \neq y$  and  that  tuples $\ve M $ and $\ve {M'} $ have different  lengths; ~~~~~~~~ \xx{\rc{\large as before $\checkmark$}}
\xx{\rc{notice variable $x$ corresponds to an observable output, which emits the linear input capability of an `address' name, so each element in $\ve{M}$ can be retrieved. }}

\item \label{i:inverse:tp}
the abstraction and  variable contexts  of $\encodingm{\,}{}{}$ have  inverse with respect to $\geq$; ~~~~~~~~ \xx{\rc{\large as before $\checkmark$}}

\item \label{con:bts:e:tp} 
$ \encom {M} \asymp \encom{\Omega} $ whenever $M$ is unsolvable of order $\infty$; ~~~~~~~~ \xx{\rc{\large Claim 2 $\checkmark$}}

\item \label{con:bts:f:tp} 
 $M$  solvable implies that the term $\encodingm{\lambda x.M}{}{}$ is unrelated by $\asymp$  to $\encodingm{\Omega}{}{}$ and to any term  of the form $\encodingm{x \ve M}{}{}$. ~~~~~~~~ \xx{\rc{\large as before $\checkmark$;}} ~~~~~~~~ \xx{\rc{notice $\encodingm{x \ve M}{}{}$ has an observable output action.}}

\end{enumerate}
\vspace{0.5em}
\end{minipage}
\end{tabular}
\noindent\rule{\textwidth}{.5pt}
\end{center}

\caption{The conditions for BT under types}\label{f:condtions:types}
\end{figure}


\finish{in Figure~\ref{f:condtions:types} added ``plain'' twice }

\subsubsection{Proof of Theorem \ref{t:resCBNlin}}
\label{sss:resCBNlin}

Now we prove Theorem \ref{t:resCBNlin}. Before beginning, we first present the encoding (\afig \ref{f:example-lazy}.a) rendered in the typed calculus $\pi^l$, as shown in \afig \ref{f:example-milner-linear}, whose design idea is explained in the course of proving the theorem.
\begin{figure}[t] 
{\small
\centering
 \begin{minipage}{0.95\textwidth}%
\centering
\[
\begin{array}{rcl}
\encom{ \lambda     x . M } & \defin &  \abs {p} \inp{ p}{x, q}.  \enco{M}{q} \\[3pt]
\encom{x} & \defin &    \abs {p} \out{x}{p}  \\[3pt]
\encom{M N } & \defin &   \abs {p} (\res{r:T_b', x:\sharp T_b}) \Big(\enco{M}{r}  \,|\,  \out r {x, p} \,|\, \\[3pt]
& &   !  \inp x q .  \enco N q \Big) \quad\mbox{(for  $x$ fresh; \xxstress{$T_b'\DEF l_\sharp (\sharp T_b, T_b), T_b\DEF \mu X. l_i (\sharp X, X)$})} 
\end{array}
\]
\end{minipage}%
}
\caption{Milner's encoding under linear typing} \label{f:example-milner-linear}
\end{figure}
\sepp

\begin{proof}[Proof of Theorem \ref{t:resCBNlin}]
\xx{(directly work on barbed congruence or (labelled) bisimulation (if any); note bisimilarity implies barbed congruence but the other direction does not hold \cite{SW01a}; perhaps try context lemma (Lemma \ref{l:context}) (ref. P.328 of \cite{SW01a}))}

We recall in \afig \ref{f:condtions:types} the soundness and completeness conditions for
BT, 
\Db adapted to typed calculi. \De As before, $\asymp$ and $\leq$ are relations on
agents. Since we are now using typed $\pi$, these relations \Db are adapted \De accordingly. We begin with some explanation and then proceed with the analysis of the conditions.
\Db
\begin{itemize}
\item To accommodate types, we designate the \emph{basic} type $T_{b}$, which is defined
  as  $T_{b} \DEF \mu X. l_i (\sharp X, X)$. When encoding a $\lambda$ term, $T_b$ is used to type the
  \emph{address} of its encoding (viz., the parameterized name of the $\lambda$ term's
  encoding). 
That is, $T_{b}$ is assumed, in the type environment, to be the type of the names used to
instantiate the address of the encoding of $\lambda$ terms. 
\De

\item 
\Db
For $\asymp$, we use $\BCTA$; this is parametric on a type environment $\Delta$ 
(assigning the type $T_{b}$ 
to the names used as
addresses in the  encoding of  $\lambda$ terms).
\De
  \xx{ (should it be parametric on a type
  environment $\Delta$?` Seems so done!)}; 
\item For $\leq$, we reuse the usual expansion $\expa$.
\item The original encoding (\afig \ref{f:example-lazy}.a) is modified using types, as shown in \afig \ref{f:example-milner-linear}. For every $\lambda$ term $M$, its encoding $\encoding{M}{}{}$ is of the type $T_b\rightarrow\Diamond$;
\item We know that if two processes are related by a (finer) untyped behavioral equivalence, then they are also related by the typed one, e.g., $\BCTA$. 
\end{itemize}


We now analyze the satisfaction of the conditions under (linear) type environment. Those not mentioned can be done as for the untyped case.
For instance, in soundness condition \conref{con:lls:f:tp}, variable $x$ corresponds to an observable output, which emits the linear input capability of an 'address' name, so each element in $\ve{M}$ can be retrieved. Then the arguments are similar to the untyped case.
Below we proceed with a number of claims.
\begin{itemize}
\item[] \textbf{Claim 1}. If $\Gamma \Vdash p:T$ in which $T\DEF T_b$ (i.e., $T\DEF \mu X. l_i (\sharp X, X)$), then $\Gamma \Vdash \encodingm{\lambda x.\Omega}{}{}\lrangle{p} \BCTA \encodingm{\Omega}{}{}\lrangle{p}$.
\item[] \textbf{Claim 2} (\xxstress{completeness condition \conref{i:sLLk:tp} and soundness condition \conref{con:bts:e:tp}}). If $\Gamma \Vdash p:T$ in which $T\DEF T_b$ (i.e., $T\DEF \mu X. l_i (\sharp X, X)$), then $\Gamma \Vdash \encodingm{M}{}{}\lrangle{p} \BCTA \encodingm{\Omega}{}{}\lrangle{p}$ for every unsolvable $M$ of order $\infty$.
\end{itemize}

{Claim 1} is valid because in Milner's encoding (\afig \ref{f:example-lazy}.a), the
abstracted name (i.e., the address) of an encoded $\lambda$ term has essentially the
\emph{linear} type in its very first place; that is, each such address is used only in one
communication should there be an application of $\lambda$. Thus, if we set up a type
environment that stipulates precisely the linearity of the address name, say $p$, an
observer from outside will at most be able to obtain the linear output capability of $p$.
Therefore, even if the input prefix in $\encodingm{\lambda x.\Omega}{}{}$ would be
observed by providing an output like $\out p {\seq{d}}$, the resulting process on the
other side, i.e., $\encodingm{\Omega}{}{}\lrangle{p} \para \out p {\seq{d}}$, would not be
observable at all, because $p$ has been exhausted in its capabilities. Since a $\lambda$
term may spawn local addresses, the linear type of $p$ has to be a recursive 
\Db
type.
\De 
{Claim 2} can be similarly analyzed.

\Xb
\Xe
\Db\finish{i do not understand the sentence above: what are exactly the changes made because
  of that claim?}  
\De

\begin{itemize}
\item[] \textbf{Claim 3} (\xxstress{completeness condition \conref{i:sLLh:tp} and soundness condition \conref{con:lls:e:tp}}). If $M $ is an unsolvable of order $0$ and $\Gamma \Vdash p:T_b$, then $\Gamma \Vdash \encom M \lrangle{p} \BCTA \encom \Omega \lrangle{p}$.
\end{itemize}
{Claim 3} is true because in that case, $\encom M \lrangle{p} \approx \encom \Omega \lrangle{p}$, where we recall $\approx$ is the bisimilarity.

\begin{itemize}
\item[] \textbf{Claim 4} (\xxstress{completeness condition \conref{i:sLLf:tp}}). Assume a type environment $\Gamma$. Then (under $\Gamma$) $\BCTA$ validates the \uptoE\ technique.
\end{itemize}
{Claim 4} states somewhat that the up-to technique can be transplanted to the typed case.
 \xx{[\bc{analyzing Claim 4}]} This result is like that for bisimulation from \cite{SW01a}. The part the expansion plays is similar. We thus sketch the argument.
We first recall Definition \ref{d:uniCAsyn} below.

\begin{flushleft}
{\small
\begin{tabular}{l}
\begin{minipage}{0.97\textwidth}
Definition \ref{d:uniCAsyn}. ~ Relation  $\asymp$ {\em validates the \uptoc\ technique} if
for any symmetric relation  $\R$ on $\pi$-processes
we have $\R \subseteq {\asymp}$ whenever
for any pair $(P,Q)\in \R $, if  $P \stm \mu P'$ then  $Q \Arcap \mu
Q'$ and there  are processes $ \til P,  \til Q$ and
a  context  $\qct$  such  that
  $P' \geq  \ct {\til P} $, $Q' \geq  \ct {\til Q}$,
and, if $n\geqNAT 0 $ is the length of the tuples $\til P$ and $\til Q$,
 at least one of the following two statements is true, for each $i \leqNAT n$:
(1) $P_i \asymp Q_i$;
(2) $P_i \RR Q_i$ and,
 if $[\cdot]_i$ occurs under an input in $C$,
also $P_i\sigma \RR  Q_i\sigma$  for all substitutions  $\sigma$.
\end{minipage}
\end{tabular}
}
\end{flushleft}
\sepp

Let $\R$ be  a relation as in Definition~\ref{d:uniCAsyn}, where $\asymp $ is $\BCTA$ and $\leq$ is $\expa$. 
Define relation $\mathcal{S}$ as below.
\[
{\mathcal{S}} \;\defin\; \asymp \,\cup\, \{(P_1,P_2) \, |\, P_i \contr \ct{\til{P_i}}
(i=1,2) \mbox{ and } (\til {P_1}, \til{P_2}) \in \R \,\cup\!\asymp
 \}
\]
Note $(\til {P_1}, \til{P_2}) \in \mathcal{R}$ stands for $(P_1^k, P_2^k) \in \mathcal{R}$ for all $P_i^k\in \til {P_i}$($i=1,2$), $k\leqslant m$ and $m$ is the number of holes in $C$.
Obviously $\R \subseteq \mathcal{S}$, so
\DSb
we have to show that $\mathcal{S} \subseteq \asymp$ (that is,
$\mathcal{S} \subseteq \BCTA$); one exploits the
characterization of
$\BCTA$ as $ \CTTA$ given in  the Context Lemma \ref{l:context}.
\DSe
The argument is routine analysis, 
except that the context $C$ and relevant processes in $\mathcal{S}$ should be well-typed, 
\Db according to \De
 the type environment $\Gamma$ designated upon $\BCTA$.
Yet since type information does not have any effect on reductions, the
analysis is similar to the untyped
\DSb
case (e.g. for $\approx$);
\DSe
see also \cite{SW01a}. 
\xx{(TO add more details to the analysis of Claim 4 ? ...Seems NO...)}


\begin{itemize}
\item[] \textbf{Claim 5} (\xxstress{soundness condition \conref{con:lls:c:tp}}). If
\[
\Gamma,a:S_1,(\til{c}-\til{b}):\til{T} \Vdash (\res {\til b:\til{S}}) ( \out a {\til c} | b(r). P ) \BCTA  (\res {\til b:\til{S}}) ( \out a {\til c} | b(r).  Q)
\] where $b \in \til b \subseteq \til c $, and $a,b$ are fresh and neither of them is of linear type $l_\sharp T_2$ for some $T_2$, then (for some $T_1$)
\[
\Gamma,r:T_1 \Vdash P \BCTA Q
\]
\end{itemize}
\xx{[\bc{analyzing Claim 5}]} The cases when $a$ or $b$ is unobservable is not
possible, which is why we assume they do not have the linear type $l_\sharp
T_2$. Then the argument is similar to the case of $\approx$. That is, feed the
two processes a concurrent (well-typed) process $a(\ve{x}).\out{b}q$ (in which
$b\in \ve{x}$), and argue as for $\approx$ to obtain the equivalence between $P$
and $Q$. \qedhere

\end{proof}


\end{document}